\documentclass[a4paper,twocolumn,superscriptaddress,10pt]{quantumarticle}
\pdfoutput=1

\usepackage{amsthm}
\usepackage{amssymb}
\usepackage{amsmath}
\usepackage{amsfonts}
\usepackage{physics}
\usepackage{graphicx}
\usepackage{dcolumn}
\usepackage{bm}
\usepackage[utf8]{inputenc}
\usepackage[table]{xcolor}
\usepackage[T1]{fontenc}
\usepackage[colorlinks=true,citecolor=teal,linkcolor=teal]{hyperref}
\hypersetup{allcolors={teal}}
\usepackage{tikz}
\usepackage{lipsum}
\usepackage{epsfig}
\usepackage{epstopdf}
\usepackage{mathtools}
\usepackage{braket}
\usepackage{hyperref}
\usepackage[capitalise]{cleveref}
\usepackage{amsthm}
\usepackage{enumitem}
\usepackage{dsfont}
\usepackage{bbold}
\usepackage{eurosym}
\usepackage{diagbox}
\usepackage{comment}
\usepackage{algorithm}
\usepackage{algpseudocode}
\usepackage{enumitem} 
    \newlist{sublist}{enumerate}{1}
    \setlist[sublist,1]{label=(\alph*)}

\crefname{section}{section}{sections}
\newcommand{\ba}{\begin{aligned}}
\newcommand{\ea}{\end{aligned}}

\newcommand{\bc}{\begin{center}}
\newcommand{\ec}{\end{center}}
\newcommand{\beq}{\begin{equation}}
\newcommand{\eeq}{\end{equation}}
\newcommand{\beqq}{\begin{equation*}}
\newcommand{\eeqq}{\end{equation*}}
\newcommand{\beqa}{\begin{align}}
\newcommand{\eeqa}{\end{align}}
\newcommand{\barr}{\begin{array}}
\newcommand{\earr}{\end{array}}
\newcommand{\bi}{\begin{itemize}}
\newcommand{\ei}{\end{itemize}}

\newcommand{\cH}{\mathcal H}

\theoremstyle{plain}
\newcounter{tho}
\newtheorem{theo}[tho]{Theorem}
\newtheorem{lemma}[tho]{Lemma}
\newtheorem{cor}[tho]{Corollary}

\newtheorem{defi}[tho]{Definition}

\footnotetext{These authors contributed equally to this work}

\begin{document}

 \title{Witnessing quantum non-Gaussianity with a single quadrature}

\author{Clara Wassner$^*$}
\affiliation{Dahlem Center for Complex Quantum Systems, Freie Universit\"at, 14195 Berlin, Germany}

\author{Jack Davis$^*$}
\affiliation{QAT team, DIENS, \'Ecole Normale Sup\'erieure, PSL University, CNRS, INRIA, 45 rue d'Ulm, Paris 75005, France}

\author{Sacha Cerf}
\affiliation{QAT team, DIENS, \'Ecole Normale Sup\'erieure, PSL University, CNRS, INRIA, 45 rue d'Ulm, Paris 75005, France}

\author{Ulysse Chabaud}
\affiliation{QAT team, DIENS, \'Ecole Normale Sup\'erieure, PSL University, CNRS, INRIA, 45 rue d'Ulm, Paris 75005, France}

\author{Francesco Arzani}
\affiliation{QAT team, DIENS, \'Ecole Normale Sup\'erieure, PSL University, CNRS, INRIA, 45 rue d'Ulm, Paris 75005, France}

\date{2026-07-06}

\begin{abstract}
    Full reconstruction of quantum states from measurement samples is often a prohibitively complex task, both in terms of the experimental setup and the scaling of the sample size with the system.  
    This motivates the relatively easier task of certifying application-specific quantities using measurements that are not tomographically complete, i.e.\ that provide only partial information about the state related to the application of interest.
    Here, we focus on simplifying the measurements needed to certify non-Gaussianity in bosonic systems, a key resource related to quantum advantage in various information processing tasks. We show that the statistics of a single quadrature measurement, corresponding to standard homodyne detection in quantum optics, can witness arbitrary degrees of non-Gaussianity as quantified by stellar rank.  Our results are based on a version of Hudson's theorem for wavefunctions, proved in a companion paper \cite{zeros}, revealing that the zeros in a homodyne distribution are signatures of quantum non-Gaussianity and higher stellar ranks. The validity of our witnesses is supported by a technical result showing that sets of states with bounded energy and finite stellar rank are compact. We provide an analysis of sample complexity, noise robustness, and experimental prospects.
    Our work simplifies the setup required to detect quantum non-Gaussianity in bosonic quantum states.
\end{abstract}

\maketitle


\paragraph{Introduction.} Quantum information processing holds the promise of significant advantages over classical methods. Bosonic systems in particular have gained prominence in several areas such as quantum sensing, quantum communication, quantum cryptography, quantum simulation and quantum computing. These systems include photonics, superconducting cavities, motional degrees of freedom of trapped ions, and optomechanical oscillators.  The quantum states of bosonic systems can be divided into two categories depending on their measurement statistics: Gaussian and non-Gaussian.  Gaussian states are well-understood \cite{Lloyd2012}, however non-Gaussianity underlies many quantum advantages \cite{fiuravsek2002gaussian,eisert2002distilling,giedke2002characterization,wenger2003maximal,garcia2004proposal,adesso2009optimal,niset2009no,frigerio2025joint}, and quantum non-Gaussianity in particular (i.e.\ ruling out mixtures of Gaussian states) is a necessary resource for any quantum computational advantage of bosonic systems compared to classical computers \cite{bartlettCVGottesmanKnill,eisertMariPosWF}. In the following we use non-Gaussianity to mean quantum non-Gaussianity.

For both fundamental physics and the development of quantum technologies, one often wants to ascertain whether a quantum system contains the necessary \textit{resources} for beyond-classical effects.
This can be done using quantum tomography \cite{dariano2003quantum}, which for quantum states of light is often performed via measurements of the quadratures of the electromagnetic field using homodyne or double-homodyne detection~\cite{LEONHARDT199589,RevModPhysTomo}.
The number of required measurements however scales exponentially with the number of subsystems.
Moreover, even under an energy constraint, bosonic state tomography is exponentially harder than that of finite-dimensional systems in terms of how many samples are needed to achieve a fixed confidence as a function of the number of subsystems \cite{mele2024learningquantumstatescontinuous}. 

On the other hand, tomography is not always a necessary burden: in many cases the relevant quantum feature of a physical system may be witnessed directly without requiring an exhaustive reconstruction.  For bosonic systems, previous approaches to witnessing non-Gaussianity include Refs.~\cite{mari2011directly,fiuravsek2013witnessing, chabaud2021witnessing, filip2011detecting, genoni2013detecting, hughes2014quantum, Park_testing_in_phase_space_2015_PRL, Park_Lu_single_quad_2017, Happ_Efremov_Nha_Schleich_2018, Kuhn_Vogel_nonstandard_filters_2018, walschaers2021non, chabaud2021certification, Fiurasek_Lachman_Filip_2021, fiuravsek2022efficient, kalash2025certifying, fiuravsek2025certification, racz2025nonGaussianity_moments}, with some focusing on specific measures such as negativity of the Wigner function \cite{mari2011directly,fiuravsek2013witnessing,chabaud2021witnessing} or the stellar rank \cite{chabaud2021certification,walschaers2021non,Fiurasek_Lachman_Filip_2021,fiuravsek2022efficient}.  The latter provides an operational ranking of non-Gaussian states related to the number of photon-addition and subtraction operations required for state preparation \cite{chabaud_stellar_2020,chabaud2022holomorphic}, and which bounds the computational usefulness of bosonic quantum states~\cite{chabaud2023resources}.
A recurring theme in many previous approaches, as generally formulated in \cite{hughes2014quantum}, is to optimize a particular measurement setting over the subset of Gaussian states with bounded mean photon number.  Such approaches however often make use of measurement setups that are tomographically complete~\cite{genoni2013detecting, hughes2014quantum, chabaud2021certification}
or require the measurement itself to be non-Gaussian~\cite{filip2011detecting, chabaud2021witnessing}.  This multitude of techniques, while successful in their appropriate regimes, raises the important question: \textit{What is the simplest possible measurement setting that can witness non-Gaussianity?}  Here we take a different approach and show that non-Gaussianity can be witnessed via samples from \textit{a single quadrature measurement}, such as the position-like quadrature $\hat x$. 

From a theoretical standpoint, we build upon a Hudson theorem for the wavefunction, proved in a companion paper \cite{zeros}, which implies that the zeros of quadrature distributions are signatures of non-Gaussianity.
Here we exploit this conceptual insight to construct witnesses with concrete operational consequences.
In particular, we demonstrate that there exist states of arbitrarily high non-Gaussianity, as quantified by stellar rank, which can be distinguished by our witnesses from any state of smaller stellar rank with the same energy. The key ingredient is a proof that the sets of states with bounded energy and finite stellar rank are compact in the trace norm, which may be of independent interest.

From an experimental perspective, while full quantum homodyne tomography requires measuring all rotated quadratures,
our witnesses 
utilize only a single quadrature angle, corresponding to a fixed measurement setup.  We include an analysis of the sample complexity of our protocol, a discussion of experimental prospects, and numerically explore how utilizing more quadrature angles makes our witnesses more robust to noise.  In particular, if, due to experimental imperfections, data from a single quadrature is insufficient, we show that adding just a single additional quadrature dramatically increases robustness.  



\paragraph{Preliminaries.}
\label{sec:preliminaries} 

Throughout, we consider the Hilbert space $\mathcal H = L^2(\mathbb R)$ of a single bosonic mode. The quadrature operators are $\hat q_\theta = e^{-i \theta \hat n }\hat x e^{i \theta \hat n } = \hat x \, \cos \theta  + \hat p \, \sin \theta$ for all $\theta\in[0,2\pi)$, where $\hat x$, $\hat p$, $\hat n = \hat a^\dagger \hat a$ are the position, momentum, and number operators, respectively.  The Fock basis, consisting of eigenstates of the number operator, and the coherent state of amplitude $\alpha$ are labeled $\{\ket{n}\}_{n=0}^\infty$ and $\ket{\alpha} = e^{-|\alpha|^2/2}\sum_{n=0}^\infty \frac{\alpha^n}{\sqrt{n!}}\ket{n}$, respectively. The \textit{stellar rank} $r^\star(\ket{\psi}) \in \mathbb N_0 \cup \{\infty\}$ of a pure state $\ket{\psi} \in \mathcal H$ is the number of zeroes, counted with multiplicity, of the stellar function (Segal-Bargmann transform) $F^\star_{\ket{\psi}}(\alpha) := e^{|\alpha|^2/2}\bra{\alpha^*}\psi\rangle$, where $\alpha \in \mathbb C$ \cite{chabaud_stellar_2020}.  Any single-mode state $\ket{\psi_r}$ of stellar rank $r<\infty$ can be parameterized by $\ket{\psi_r} = \hat D(\alpha) \hat S(\chi)\sum_{n =0}^r c_n \ket{n}$, where $\hat S(\chi)= e^{\frac{1}{2}(\chi^* \hat a^2 - \chi  \hat a^{\dagger^2})} $ is the squeezing operator, $\hat D(\alpha)= e^{\alpha \hat a ^\dagger - \alpha^*\hat a}$ is the displacement operator, $\sum_{n=0}^r |c_n|^2=1$, and $\alpha,\chi,c_n\in\mathbb C$ \cite{lutkenhaus1995nonclassical, chabaud_stellar_2020}.  A pure state is Gaussian if and only if it has stellar rank zero, and such states are identified with the displaced squeezed coherent states \cite{Yuen_squeezed_1976, V_V_Dodonov_2002}.

We also quote a key result from \cite{zeros}, which shows that any quadrature wavefunction (i.e.\ the expansion coefficients over the eigenbasis $\{\ket{q}_{\hat q_\theta}\}_q$ of $\hat q_\theta$) can be extended to a holomorphic function on the whole complex plane, provided the state satisfies an energy bound:

\begin{theo}[Hudson theorem for the wavefunction \cite{zeros}]\label{th:Hudson}
    For a pure state $\ket\psi$, assume there exists an $s>1$ such that $\ket\psi$ satisfies the energy bound $\langle s^{\hat n}\rangle_\psi<+\infty$. Then $\ket\psi$ is non-Gaussian if and only if all of its quadrature wavefunctions have a zero over $\mathbb C$.
\end{theo}

\noindent 
This implies that for most pure states the complex zeros of the wavefunction may be thought of as signatures of non-Gaussianity.  When these zeros are over the real line, they are also zeros of the corresponding quadrature distribution. This suggests a strategy for witnessing non-Gaussianity by measuring a quadrature and looking at the zeros of the outcome distribution. 


\paragraph{Witnessing non-Gaussianity.}
\label{sec:witness}
%
%
Consider the following problem: having only access to samples from a single quadrature distribution $p_{\hat\rho,\theta}(q)\coloneqq\bra q\hat\rho\ket q_{\hat q_\theta}$ of some state $\hat \rho$ for some fixed $\theta$, how can we detect that $\hat \rho$ is non-Gaussian?
First note that one cannot witness non-Gaussianity from a single quadrature distribution for arbitrary states: the state $\int p_{\hat\rho,\theta}(q) \ket{q}\!\bra{q}_{\hat q_\theta}\mathrm dq$ yields the same statistics under the quadrature measurement $\hat q_\theta$ as the state $\hat \rho$, but is (the limit of) a convex mixture of Gaussian states.  Such a mixture however is unphysical as it requires arbitrarily localized states with infinite energy.  This leads us to consider statistical mixtures of energy-constrained pure states:
\begin{equation}
    \mathcal S^E \coloneqq\mathrm{cch}(\{\psi\in\mathcal H\,|\,\langle\hat n\rangle\le E\}),
\end{equation}
where the closed convex hull $\mathrm{cch}(S)$ of a set $S$ is the closure of the set of convex combinations of elements in $S$. In other words, we consider mixed states for which there exists a spectral decomposition where none of the pure states involved has a mean energy exceeding $E$. We further consider subsets of $\mathcal S^E$ with finite stellar rank:
\begin{defi}[Set of energy bounded, finite stellar rank states]
    The set of states with stellar rank at most $r$ and energy bounded by $E$, in the sense that $\hat \rho \in  \mathcal S^E$, is denoted $\mathcal S^E_r$. 
\end{defi}

We may now introduce our key definition:

\begin{defi}[Witness and threshold value]\label{def:witness}
    Let $x\in\mathbb R$, $\theta\in[0,2\pi)$, $\eta >0$ and $E\ge0$. 
    We define a non-Gaussianity witness as the quadrature projector onto a window of size $\eta $ around $x$:
    \begin{equation}\label{eq:witness_def}
        \hat W_{\theta,x,\eta} \coloneqq \int_{x-\frac\eta 2}^{x+\frac\eta 2} \ket{q}\!\bra{q}_{\hat q_\theta} \mathrm d q.
    \end{equation}
    This witness operator satisfies $\mathbb 0\preceq\hat W_{\theta,x,\eta}\preceq \mathbb 1$, and its threshold value at energy $E$ is defined by 
    \begin{equation}
        \label{eq:threshold_value_def}
        w_{\theta,x,\eta}^E\coloneqq\inf_{\hat\sigma \in \mathcal S_{0}^E}  \Tr(\hat\sigma\hat W_{\theta,x,\eta}).
    \end{equation}
\end{defi}

By definition, measuring the witness operator $\hat W_{\theta,x,\eta}$ on copies of an unknown state $\hat\rho\in\mathcal S^E$ and recording an expectation value smaller than $w_{\theta,x,\eta}^E$ certifies that $\hat\rho\notin\mathcal S^E_0$, i.e.~$\hat\rho$ is non-Gaussian. This procedure amounts to estimating the magnitude of a quadrature probability distribution at angle $\theta$ around a specific point $x$. The intuition behind this definition is based on \cref{th:Hudson}: since zeros in the quadrature wavefunctions (and thus in quadrature distributions) are signatures of non-Gaussianity, only specific non-Gaussian states will be able to achieve a minimal value for the witness.

This witnessing strategy requires that the witness is sound, i.e.~its threshold value is non-zero.

\begin{theo}[Soundness of witness]
\label{th:Wsound}
Let $x\in\mathbb R$ and $E\ge0$. There exists $\delta>0$ such that for all $0< \eta \le\delta$,
    \begin{equation}\label{eq:Wsound}
        w_{\theta,x,\eta}^E>0.
    \end{equation}
\end{theo}
Moreover, for a window size small enough, the witness is complete on the set of energy-bounded states whose quadrature distribution vanishes somewhere:
\begin{theo}[Completeness of witness]
\label{th:Wcorrect}
    Let $x\in\mathbb R$ and $E\ge0$, let $\theta\in[0,2\pi)$, and let $\hat\rho\in\mathcal S^E$ be a state with quadrature distribution $p_{\hat\rho,\theta}(q)=\bra q\hat\rho\ket q_{\hat q_\theta}$ such that $p_{\hat\rho,\theta}(x)=0$. Then, there exists $\delta>0$ such that for all $\eta \le\delta$,
    \begin{equation}\label{eq:Wcorrect}
        \Tr(\hat\rho\hat W_{\theta,x,\eta})<w_{\theta,x,\eta}^E.
    \end{equation}
\end{theo}

The proofs of \cref{th:Wsound} and \cref{th:Wcorrect} rely on the following technical result which may be of independent interest:
\begin{lemma}
    Let $r\in\mathbb N$ and $E\ge0$. The set $\mathcal S_r^E$ is compact in the trace norm. 
      \label{th:compactness_lemma}
\end{lemma}
The following Corollary of \cref{th:Wcorrect} allows for an operational interpretation of the witness violation.

\begin{cor}[Operational witness]
 Let $x\in\mathbb R$, $\eta >0$ and $E\ge0$. For any state $\hat\rho$ the amount of violation for the witness $\hat W_{\theta,x,\eta}$ satisfies
    \begin{equation}\label{eq:Woperational}
        \langle \hat \Delta^E_{\theta,x,\eta } \rangle \equiv w_{\theta,x,\eta}^E-\Tr(\hat\rho\hat W_{\theta,x,\eta})\le \inf_{\hat\sigma\in\mathcal S_0^E} D(\hat\rho,\hat\sigma),
    \end{equation}
    where $D$ denotes the trace distance.
    \label{lem:operational_witness}
\end{cor}
The proofs of Lemma~\ref{th:compactness_lemma} (\cref{app:comp_lem_proof}) and \cref{th:Wsound}, \cref{th:Wcorrect}, and \cref{lem:operational_witness} (all \cref{app:valid_wit}) are deferred to the Appendix.

We end this section by mentioning two notable generalizations of our witnesses, detailed in the Appendix. First, we generalize the witness definition to include Gaussian-blurred homodyne measurements (\cref{sec:Gaussian_windows}), leading to convenient expressions for threshold values. Second, we provide a straightforward adaptation of the above results for witnessing a specific stellar rank $k$ by optimizing over $\mathcal S_{k-1}^E$ (\cref{sec:witness_stellar}).




\paragraph{Experimental protocol.}


As previously stated, to witness non-Gaussianity from a single quadrature dataset there must be an energy bound on the target state, i.e.\ that $\hat\rho \in \mathcal S^E$.  It should be possible to estimate $E$ from the experimental conditions (e.g.\ the preparation procedure, independent energy measurements, etc.).  It follows that the typical conclusion resulting from a successful application of our method will be ``either the state is non-Gaussian, or it has energy higher than $E$''.  The operational witnessing procedure is summarized in Table~\ref{tab:pseudo_code_witn}.
\begin{table*}[t]
\centering
\caption{\textbf{Witnessing non-Gaussianity from a single quadrature}}
\label{tab:pseudo_code_witn}
\begin{tabular}{ll}
\hline
\textbf{Given:} & Copies of $\hat\rho$ with the promise that $\hat\rho \in \mathcal S^E$ \\ \hline
\textbf{Step 1:} & Identify a quadrature $\hat q_\theta$ and a point $x$ such that the probability distribution of $\hat\rho$ \\
&is close to zero at $x$. \\
\textbf{Step 2:} & Fix an $\eta > 0$ and define witness $\hat W_{\theta,x,\eta}$~\eqref{eq:witness_def}. \\
\textbf{Step 3:} & Determine threshold value $w^E_{\theta,x,\eta}$ via numerical optimization over $\mathcal S^E_0$. \\
\textbf{Step 4:} & Determine expectation value $\smash{\mathrm{Tr}(\hat\rho \hat W_{\theta,x,\eta})}$ experimentally by measuring $\hat q_\theta$. \\
\textbf{Step 5:} & Compare $\smash{\mathrm{Tr}(\hat\rho \hat W_{\theta,x,\eta})}$ with $w^E_{\theta,x,\eta}$: \\
& 
$\mathrm{Tr}(\hat\rho \hat W_{\theta,x,\eta}) < w^E_{\theta,x,\eta}$ $\Rightarrow\ \hat\rho$ \text{non-Gaussian (or } $\hat\rho \notin \mathcal S^E)$. \\
& 
$\mathrm{Tr}(\hat\rho \hat W_{\theta,x,\eta}) \geq w^E_{\theta,x,\eta}$  $\Rightarrow$ \text{no conclusion possible.}
\\ \hline
\end{tabular}
\end{table*}

First, we have to identify a quadrature angle $\theta$ such that the quadrature distribution $\bra q\hat\rho\ket q_{\hat q_\theta}$ admits a zero at $q_\theta = x$. This can be done by exploiting theoretical knowledge about the prepared state $\hat\rho$ or by looking for a suitable quadrature experimentally. 
Strictly speaking, in the second approach we leave the realm of only measuring a single quadrature.  
But this might be anyway necessary to increase the robustness of the protocol as we discuss shortly. 
While in general our witnesses require knowledge of the right quadratures to certify a given state, for states diagonal in the Fock basis, phase-randomized homodyne detection yields the phase-averaged quadrature distribution. In particular, for an ideal Fock state this distribution is identical to that of any fixed quadrature, so our methods can be applied to certify its stellar rank also using a freely drifting (non-phase-locked) local oscillator. Furthermore, the interval containing the zero of the observed state can be inferred from measured data. Specifically, in an experimental search for a pair $(\theta,x)$, the acquired quadrature samples for each angle can be divided into two batches. The first batch is used to reconstruct the quadrature probability density. Based on this reconstruction, one selects the center $x$ of a window enclosing a dip in the distribution. The second batch is then used in the subsequent witnessing protocol. An example illustrating this is presented in \cref{sec:heuristic_witness_construction}. 

Having identified a pair $(\theta, x)$, the witness observable $\hat W_{\theta,x,\eta}$ is defined according to Eq.~\eqref{eq:witness_def} by fixing a window width $\eta > 0$. This can be found heuristically but an appropriate width to witness a violation in the presence of a zero is guaranteed to exist (see \cref{sec:theo_guar_bin_size}).

We next determine the threshold $w^E_{\theta ,x,\eta }$ \eqref{eq:threshold_value_def}.  The minimization over energy-bounded Gaussians is performed via numerical optimization, and the linearity of the witness implies the minimum is attained for a pure state. We provide parametrized expressions for the energy and quadrature distributions of any pure state with finite stellar rank in \cref{sec:explicit_expressions_threshold}. The optimization is numerically straightforward but finding thresholds to certify higher stellar ranks becomes difficult due to the increasing number of parameters and the nonlinearity of the objective function in these parameters.  

Finally, the expectation $\smash{\text{Tr} (\hat\rho \hat W_{\theta,x,\eta})}$ is experimentally estimated and compared with the threshold. 
If $\smash{\text{Tr} (\hat\rho \hat W_{\theta,x,\eta})} < w^E_{\theta ,x,\eta }$ then, under the energy promise, the prepared state is non-Gaussian.  Moreover, its trace distance to $\mathcal S_0^E$ is lower-bounded by $w^E_{\theta ,x,\eta }-\smash{\text{Tr} (\hat\rho \hat W_{\theta,x,\eta})}$ via Lemma~\ref{lem:operational_witness}.  If there is no violation then one can try repeating the procedure with a smaller window size $\eta$. 

The question of which target states are most compatible with this procedure amounts to determining which states contain at least one quadrature distribution zero; see \cite{zeros}. 
We point out that the Fock state $\ket{n}$ performs particularly well since it has $n$ distinct zeros in every quadrature distribution.
In summary, our single-quadrature witnessing protocol assumes a known energy constraint and some prior knowledge of the prepared state, as is typical in quantum state-preparation experiments. Under these conditions, non-Gaussianity can be witnessed from measurements of a single quadrature.

\paragraph{Sample complexity.} 
\label{sec:sample_complexity}

,
In practice, any experiment 
 
will of course rely on a finite number of measurement samples. 

Another problem is a trivial ``cheating'' strategy: center the witness outside the range of observed outcomes. Both issues are addressed by the following result:
\begin{theo}[Sample complexity for witnessing non-Gaussianity]\label{th:sampleComplexity1}
    Consider a witness $\hat{W}_{\theta,x,\eta}$ with threshold $w_{\theta,x,\eta}^E$ and suppose the state being measured is in $\mathcal{S}_0^E$. The probability that the 
    witness estimate
    derived from $M$ measurement outcomes, $p_\mathrm{fail}$, will be smaller than a fixed $w_{\theta,x,\eta}^E-\epsilon$, with $\epsilon>0$, is bounded by \begin{equation}
        p_\mathrm{fail}\leq \exp\left( -2M\epsilon ^2\right).
    \end{equation}
\end{theo}

\begin{proof}
    Consider a witness $\hat{W}_{\theta,x,\eta }$ \eqref{eq:witness_def} with threshold value $w_{\theta,x,\eta}^E$ \eqref{eq:threshold_value_def}.  For the prepared state $\hat\sigma$ define the 
    estimate
    $\bar{w}_{\hat\sigma} = \frac{1}{  M}\sum_{j=1}^{M}\chi_{\left[x-\frac{\eta }{2},x+\frac{\eta }{2}\right]}\left(q_j\right)$, where $\chi_{\left[a,b\right]}$ is an indicator function and $\{q_1,\ldots,q_M\}$ are the $M$ measurement samples of $\hat{q}_\theta$. By definition, $\bar{w}_{\hat\sigma}$ is a sum of i.i.d.\ random variables with range $\left[0,1\right]$.
    For any $\epsilon >0$, we have 
    \begin{equation}
        \begin{aligned}
            p_\mathrm{fail}&=\mathrm{Pr}\left(\bar{w}_{\hat\sigma}<w_{\theta,x,\eta}^E-\epsilon \right) \\& \leq \mathrm{Pr}\left(\bar{w}_{\hat\sigma}<\mathbb{E}\left[\bar{w}_{\hat\sigma}\right]-\epsilon \right)\\
            &\leq\exp\left(-2M\epsilon ^2 \right)
        \end{aligned}
    \end{equation}
    where the first inequality follows from the definition of $\hat\sigma$ and $w_{\theta,x,\eta}^E$ and the second from Hoeffding's inequality~\cite{hoeffding1963probability}.
\end{proof} 

In other words, witnessing a value $\bar{w}_{\hat\sigma}<w_{\theta,x,\eta}^E-\epsilon $ certifies that $\hat\sigma$ is non-Gaussian (or has energy larger than $E$) with confidence $1-\delta$, provided $\bar{w}_{\hat\sigma}$ is estimated using at least $\log(1/\delta)/(2\epsilon ^2)$ measurement samples. This also rules out the aforementioned cheating strategy because obtaining outcomes at a distant $x$ is extremely unlikely for energy-bounded Gaussians, thereby requiring a prohibitively high sample complexity to obtain any reasonable degree of confidence.


\paragraph{Robustness of the protocol.}
\label{sec:robustness}
 
In this section we explore how the choice of window size $\eta$ combined with experimental imperfections might impact the ability to witness non-Gaussianity.  As a prototypical example, consider as a target the single-photon state $\ket{1}$ with energy $E=1$. This state is rotation-symmetric, implying that all quadrature distributions $|\langle q_\theta | 1 \rangle|^2 = \frac{2}{\sqrt{\pi }} q_\theta^2 e^{-q_\theta^2}$ have the same zero-set: $q_\theta=0$.  Fixing an angle $\theta$, Theorem \ref{th:Wcorrect} implies a set of windows around the origin within which no Gaussian in $\mathcal S_0^1$ can minimize its probability profile better than $\ket{1}$. The optimal Gaussian performance for a general window size $\eta$ (i.e.\ the threshold $w^1_{\theta, 0, \eta}$) is found numerically using analytic expressions presented in \cref{sec:explicit_expressions_threshold}. 

We then take the difference between the target's witness expectation and the optimal Gaussian expectation to obtain the maximal possible violation $\langle \hat \Delta^E_{\theta,x,\eta } \rangle$ \eqref{eq:Woperational}.  The solid blue curve in Figs.\ \ref{fig:stellar-rank-1-width-and-num-quads} and \ref{fig:minimal_info_mins} show this violation as a function of $\eta$. A value above zero thus indicates the certifiability of the energy-bounded prepared state being non-Gaussian.

\begin{figure}[h!]
    \centering
    \includegraphics[width=\linewidth]{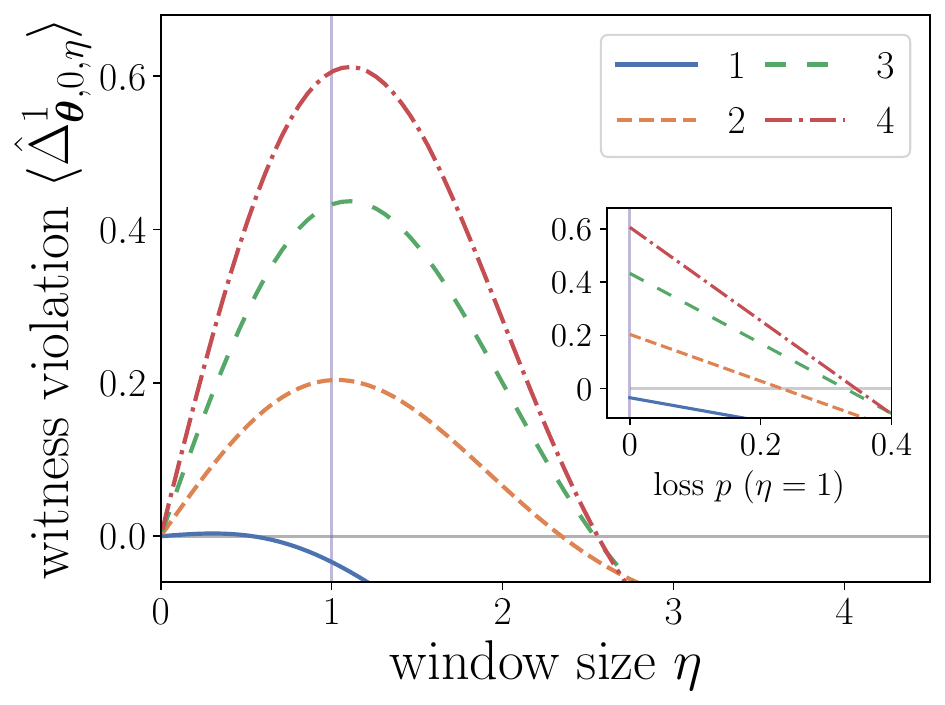}
    \caption{Outer figure: Optimized violation over $\mathcal S_0^1$ of the non-Gaussianity witness, $\langle \hat \Delta^1_{\bm\theta, 0,\eta }\rangle$, as a function of window size and number of equiangularly distributed quadratures used.  The single-photon state $\ket{1}$ is the target.  Inset: The same optimized violation with the horizontal axis denoting the loss $p$ (at a fixed window $\eta = 1$) of a decohered target state.}
    \label{fig:stellar-rank-1-width-and-num-quads}
\end{figure} 

The additional curves in Fig.\ \ref{fig:stellar-rank-1-width-and-num-quads} exploit the rotational symmetry of the target state and are found by minimizing $\sum_{k=0}^{n-1} \hat W_{k \frac{\pi}{n},0,\eta}$ over $S_0^1$ and subtracting $n$ copies of the target expectation.  See \cref{sec:sampleComplexityManyQuads} for a discussion of the sample complexity in this modified measurement set-up.  The inset shows how the optimal Gaussian performance linearly increases relative to a lossy single-photon target state $\hat \rho(p) = p \ketbra{0}{0} + (1-p) \ketbra{1}{1}$.  These two results show that while using a single quadrature for witnessing non-Gaussianity is only achievable for near-perfect target states, the inclusion of even one additional quadrature yields roughly a hundredfold stronger witness violation and therefore a much higher robustness to loss. Since loss may in practice make it necessary to consider witnesses based on more than a single quadrature measurement, we explore in \cref{sec:no_energy_promise} an extended witness defined as a convex combination of the number operator and the original witness operator. The advantage of this approach is that it removes the need for an energy promise on the target state. In the original approach, this energy promise can be verified from measurements of three linearly independent quadratures.

\begin{figure}
    \centering
\includegraphics[width=\linewidth]{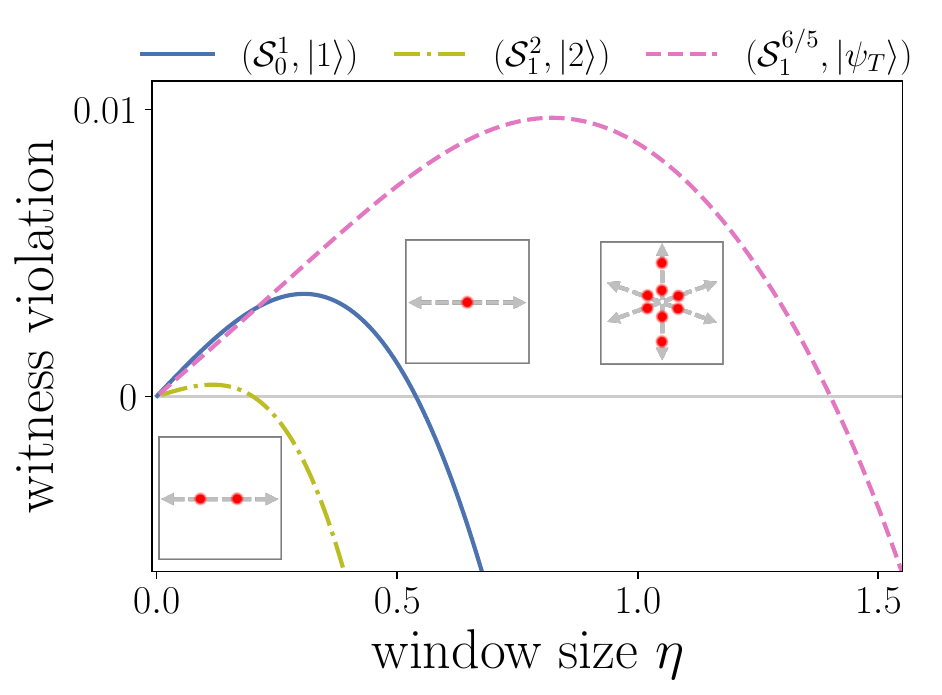}
    \caption{Optimized witness violation for three set-ups that do \textit{not} exploit any symmetry of the target state: the solid blue curve (same as Fig.\ \ref{fig:stellar-rank-1-width-and-num-quads}) denotes the violation over $\mathcal S_0^1$ using one quadrature, the dot-dashed olive curve denotes the violation over $\mathcal S_1^2$ using one quadrature, and the dashed pink curve denotes the violation over $\mathcal S_1^{6/5}$ using six quadratures.  The insets schematically display the zeros and their quadratures. }
    \label{fig:minimal_info_mins}
\end{figure}

Generally, using more quadratures also leads to improvements if we wish to witness specific stellar ranks, though with some additional nuance.  If the target is now the two-photon state $\ket{2}$, which has zeros at $\pm 1/\sqrt{2}$ in all quadratures, the optimal witness violation over $\mathcal S_1^2$ using a single quadrature is again non-zero but highly sensitive to noise --- see the bottom olive curve in Fig.\ \ref{fig:minimal_info_mins}.  And if more quadratures are utilized like the previous example, the performance again dramatically increases (not shown).  However, consider the target state 
$\ket{\psi_T} = \sqrt{1/10}\ket{0} + i\sqrt{3/5}\ket{1}+\sqrt{3/10}\ket{2}$.  It has energy $\frac65$, stellar rank $2$, and is less symmetric than the Fock state $\ket{2}$, with only 8 zeros distributed over 6 quadrature angles $\theta_i\in[0,2\pi)$.  Optimizing over $\mathcal S_1^{6/5}$ along any subset of these six quadratures again produces a non-trivial violation landscape that can witness stellar rank 2 --- see the pink dashed curve in Fig.\ \ref{fig:minimal_info_mins} for the case of using all six angles.  There the maximal window size $\eta \approx 1.4$ and the maximal witness violation is high relative to any single-quadrature analysis.  Unlike Fock states however, we no longer have the option to add additional angles to combat the noise sensitivity.  

That being said, if one instead optimizes over $S_0^{6/5}$ along the six quadratures then we recover strong robustness (i.e.\ $\hat \Delta \approx 0.53$ for $\eta \approx 0.83$, not shown), albeit now with respect to witnessing general non-Gaussianity rather than specifically stellar rank 2.  This result is unfortunately not universal: we have found other stellar-rank-2 target states where the energy-bounded Gaussians do not yield a high violation despite summing over all available quadrature zeros arising from the target state.  All together these results suggest a subtle relationship between energy, robustness to noise, and the placement of quadrature zeros.
In particular it appears that the most robust target states tend to have high stellar rank, low energy, and many zero-containing quadratures with zeros located near the origin.  One aspect of this relationship that we may always conclude is that the robustness is highly sensitive to the energy bound since the threshold value $w_{\theta,x,\eta}^E$ decreases exponentially with $E$.




 


\paragraph{Discussion.}
\label{sec:discussion}

Here we designed a simple and operational protocol to certify the non-Gaussianity of a state based on the zeros of quadrature probability distributions. 
We constructed a family of witness operators, the expectation value of which is lower-bounded by a computable threshold if the state is Gaussian and energy-bounded.  
We then demonstrated how to engineer the witness operator so that a threshold violation must occur if the tested state exhibits a zero on one of its quadratures.  
In a companion paper \cite{zeros} we establish the general prevalence of such zeros and conjecture they exist for all energy-bounded non-Gaussian pure states.
For some specific states, we explained how to adapt our protocol to certify arbitrarily high non-Gaussianity as measured by the stellar rank, using a single quadrature, and bounded the sample complexity of such certification.
While theoretically sound, these witnesses can only tolerate a small amount of noise in the sense that a mixture of Gaussian distributions might be able to match the witness value of a lossy non-Gaussian prepared state.  To compensate, we provided a numerical analysis that reveals that measuring even a single additional quadrature angle can drastically improve the noise robustness of the witnesses.

Our work opens several research directions.  A natural next step is to better characterize the class of target states most compatible with our witnessing protocol, both for general non-Gaussianity and for specific stellar ranks.  Moreover, it would be of theoretical and practical interest to understand the structure of the target states for which the certification is most robust.  This will require a deeper understanding of how the complex zeros of a wavefunction evolve when changing the quadrature angle.  We also note that ensuring the witness operator is centered on an exact zero is not strictly necessary: it suffices that the distribution dips low enough so that no energy-bounded Gaussian state can achieve a lower witness value. This may offer greater flexibility in optimizing the efficiency of the protocol.

Another line of research is to adapt our framework to a heterodyne measurement set-up \cite{hughes2014quantum}.  The witnesses would be two-dimensional windows of coarse-grained coherent state POVM elements, centred at Husimi function zeros.  This could have a more direct link to stellar rank, and it would be worthwhile to compare the noise robustness between approaches.


\paragraph*{Acknowledgements.}

U.C.\ thanks O.~Hahn, G.~Ferrini, A.~Ferraro, V.~Upreti and Z.~Van Herstraeten for interesting discussions. U.C.\ and J.D.\ acknowledge funding from the European Union’s Horizon Europe Framework Programme (EIC Pathfinder Challenge project Veriqub) under Grant Agreement No.~101114899.

\paragraph*{Author contribution statement.}
C.W. and J.D. led the theoretical calculations and numerical simulations, and contributed equally to this work. All authors contributed to discussing the results and to the writing of the manuscript. The generative AI tool "Claude" was used to assist with code writing for some of the numerical simulations. The authors take full responsibility for all content, including any code produced with such assistance.

\paragraph*{Data availability.}

The data that support the findings of
this article are openly available \cite{data_plotting}.


\bibliographystyle{quantum}
\bibliography{ref}   

---------

\onecolumngrid
\newpage
\appendix

\section{The sets of states with bounded energy and finite stellar rank are compact}
\label{app:comp_lem_proof}
In this section, we give a proof of Lemma 6 in the main text.

We consider here the metric space $\mathcal T(\mathcal H)$ of all trace class operators on the Hilbert space $\mathcal H = L^2(\mathbb R)$ with the trace norm. The set of quantum states is denoted as $\mathcal D(\cH) \subset \mathcal T(\mathcal H) $, i.e.~the set of linear, positive operators $\hat \rho$ on $\cH$ with a finite trace and $\|\hat \rho \|_1 = 1$.  The space $(\mathcal T(\mathcal H),\|\cdot\|_{1})$ is a Banach space and the following basics facts and known results, valid in this space, will be used in the proof.
\begin{enumerate}[label=(\roman*)]
    \item The intersection of a closed set and a compact set is compact. 
    \item The set of states with bounded average energy is compact \cite{holevoEntanglementassistedCapacityConstrained2003}.
    \item The set of pure states is closed (it is the preimage of the set $\{1\}$ under the continuous $2$-norm function).
    \item The set of states with stellar rank bounded by $r$ is closed \cite[Theorem 5]{chabaud_stellar_2020}.
    \item The closed convex hull of a compact set is compact \cite[Theorem 5.35]{aliprantisInfiniteDimensionalAnalysis2006}.
\end{enumerate}
For convenience we repeat the relevant definition of the set of statistical mixtures, or the closure of the convex hull, of energy-constrained pure states
\begin{equation}
    \mathcal S^E\coloneqq \mathrm{cch}(\{\psi\in\mathcal H\,|\,\langle\hat n\rangle_{\psi}\le E\}),
\end{equation}
and its subset of states with stellar rank bounded by $r$ which we denote by $\mathcal S_r^E$. We remark that the set $\mathcal S^E$ contains all states which can be written as a convex combination of pure states whose mean energy is individually not exceeding $E$. We want to prove that $\mathcal S_r^E$ is compact by repeatedly making use of ($\operatorname{i}$). Obviously, $\mathcal S_r^E$ is the intersection of $\mathcal S^E$ with the set of states with stellar rank bounded by $r$. Since the latter is closed ($\operatorname{iv}$) we are left with proving that the set $\mathcal S^E$ is compact. The set \begin{equation}
   \mathcal S^E_\mathrm{pure}\coloneqq  \{\psi\in\mathcal D(\mathcal H)|\|\psi^2\|_1 = 1,\,\langle\hat n\rangle_{\psi}\le E\}
\end{equation}
of pure, energy-bounded states is the intersection of a closed ($\operatorname{iii}$) and compact ($\operatorname{ii}$) set and is hence itself compact by ($\operatorname{i}$). Since by definition $\mathcal S^E$ is the closed convex hull of the compact set $\mathcal S^E_\mathrm{pure}$, $\mathcal S^E$ is by (v) compact.   
  Recalling that $\mathcal S^E_\mathrm{pure}$ is compact we use ($\operatorname{v}$) to deduce that $\mathcal S^E$ is compact. The Lemma follows. 

\section{Optimal Gaussian states are pure and have maximal energy}

In this section, we show that the optimal Gaussian states minimizing the witness expectation value for a given energy parameter $E$ are pure and have maximal energy $E$.

\begin{theo}
Let us consider the witness operator $W$ defined as:
\begin{equation}
    \hat W_{\theta,x,\eta} = \int_{x - \eta/2}^{x + \eta/2} \dyad{q}{q}_{\hat q_\theta} dq,
\end{equation}
where $\ket{q}_{\hat q_\theta}$ is the eigenstate of the quadrature operator $\hat q_\theta$ associated to the eigenvalue $q$. Then: 
\begin{equation}
    \min_{\rho \in S_0^E} \Tr(\hat W_{\theta,x,\eta}\rho) = \min_{\substack{\ket\psi\!\bra\psi\in S_0^E\\ \bra\psi \hat n\ket\psi = E}} \bra\psi \hat W_{\theta,x,\eta}\ket\psi.
\end{equation}
\end{theo}

\begin{proof}
First, notice that
\begin{equation}
    \min_{\rho \in S_0^E} \Tr(\hat W_{\theta,x,\eta}\rho) = \min_{\ket\psi\!\bra\psi\in S_0^E} \bra\psi \hat W_{\theta,x,\eta}\ket\psi,
\end{equation}
because the objective function is linear and the set $\mathcal S_0^E$ is the convex set of mixtures of pure Gaussian states with energy less than or equal to $E$.

Let $\rho=\ket\psi\!\bra\psi \in S^E_0$ be a pure Gaussian state, and $\hat D(t) = e^{-it \hat p_\theta}$ be the displacement operator along the quadrature at angle $\theta$ by a quantity $t$. 

Suppose $\Tr(\hat n \rho) < E$. Then, there exists a neighborhood $V\subset \mathbb R$ of $0$ such that for all $t \in V$, $\Tr(\hat n \hat D(t) \rho \hat D(t)^\dagger) < E$. We show here that $t = 0$ is not a local minimum of $t \mapsto \Tr( \hat n D(t)\rho D(t)^\dagger)$, which implies that $\rho$ is not a global minimum of $\Tr(\hat W_{\theta,x,\eta}\rho)$.
To find the extremum of our witness operator with respect to this displacement, we evaluate the derivative of its expectation value:
\begin{equation}
    \frac{d}{dt} \Tr(\hat W_{\theta,x,\eta} \rho(t)) \bigg|_{t=0} = \Tr \left( \hat W_{\theta,x,\eta}\frac{d}{dt} \left(e^{-it \hat p_\theta} \rho e^{it\hat p_\theta}\right) \right) \bigg|_{t=0} = -i \Tr(\hat W_{\theta,x,\eta} [\hat p_\theta, \rho]).
\end{equation}
Using the cyclic property of the trace, we rewrite this in terms of the commutator of $\hat W_{\theta,x,\eta}$ and $\hat p_\theta$:
\begin{equation}
    \frac{d}{dt} \Tr(\hat W_{\theta,x,\eta} \rho(t)) \bigg|_{t=0} = i \Tr([\hat W_{\theta,x,\eta}, \hat p_\theta] \rho).
\end{equation}

Next, we compute the commutator $[\hat W_{\theta,x,\eta}, \hat p_\theta]$. Since $\hat p_\theta$ is the generator of translations for the quadrature $\hat q_\theta$, its action on the position projectors is $i[\dyad{q}{q}_{\hat q_\theta}, \hat p_\theta] = -\frac{d}{dq} \dyad{q}{q}_{\hat q_\theta}$. Therefore:
\begin{equation}
    i[\hat W_{\theta,x,\eta}, \hat p_\theta] = - \int_{x - \eta/2}^{x + \eta/2} \frac{d}{dq} \dyad{q}{q}_{\hat q_\theta} dq.
\end{equation}
Evaluating this integral yields a difference of the projectors evaluated at the boundaries of the integration interval:
\begin{equation}
    i[\hat W_{\theta,x,\eta}, \hat p_\theta] = \dyad{(x-\eta/2)}{(x-\eta/2)}_{\hat q_\theta}  - \dyad{(x+\eta/2)}{(x+\eta/2)}_{\hat q_\theta} 
\end{equation}

Taking the trace with the density matrix $\rho$, the derivative of the expectation value becomes the difference of the probability densities at the boundaries:
\begin{equation}
    \frac{d}{dt} \Tr(\hat W_{\theta,x,\eta} \rho(t)) \bigg|_{t=0} = \bra{x-\eta/2}_{\hat q_\theta} \rho \ket{x-\eta/2}_{\hat q_\theta} - \bra{x+\eta/2}_{\hat q_\theta} \rho \ket{x+\eta/2}_{\hat q_\theta}.
\end{equation}

For this derivative to vanish—which is a necessary condition for reaching an extremum—these boundary values must be exactly equal:
\begin{equation}
    \bra{x-\eta/2}_{\hat q_\theta} \rho \ket{x-\eta/2}_{\hat q_\theta} = \bra{x+\eta/2}_{\hat q_\theta} \rho \ket{x+\eta/2}_{\hat q_\theta}.
\end{equation}

Since $\rho$ is a pure Gaussian state, its probability distribution in the $\hat q_\theta$ basis is a Gaussian function. A Gaussian function takes equal values at the boundaries of an interval if and only if the interval is exactly centered on its mean. Because a Gaussian profile strictly decreases away from its peak, centering the Gaussian exactly on the interval $[x-\eta/2, x+\eta/2]$ ensures we are maximizing the overlap with $\hat W_{\theta,x,\eta}$, rather than minimizing it. Thus, if $\rho$ is a local optimum with energy $\Tr(\hat n \rho)<E$, it is a local maximum, and hence not a global minimum.
\end{proof}

\section{Validity of the witnesses}
\label{app:valid_wit}

In this section, we show the validity of the witnesses defined in the main text. We also include at the end of this section the proof of Corollary 7 in the main text about the operational interpretation of the witness violation.  

Firstly, we show in \cref{sec:valid} that, for a given state $\rho$ of energy $E$ with a quadrature zero at $q_\theta=x$, there always exists a window size $\eta(\rho,E,x)$ for which the witness $\hat W_{\theta,x,\eta}$ successfully distinguishes the state $\rho$ from any Gaussian state of energy at most $E$.

Secondly, we strengthen this result in \cref{sec:theo_guar_bin_size} by showing that, for a given state $\rho$ satisfying the energy constraint $\langle s^{\hat n}\rangle_\rho\le S$ for some $s>1$ with a quadrature zero at $q_\theta=x$, there always exists a window size $\eta(x_0,s,S,E)$ for which the witness $\hat W_{\theta,x_0,\eta}$ successfully distinguishes the state $\rho$ from any Gaussian state of energy at most $E$. Importantly, in that case, the bound on the window size does not depend on the specific state at hand and only depends on the $x_0$ and the energy parameters.

\subsection{Completeness and soundness}
\label{sec:valid}

 For convenience, we recall the relevant definitions: Let $x\in\mathbb R$, $\theta\in[0,2\pi)$, $\eta >0$ and $E\ge0$. 
    We define a non-Gaussianity witness as the quadrature projector onto a window of size $\eta $ around $x$:
    \begin{equation}\label{eq:witness_def_app}
        \hat W_{\theta,x,\eta} \coloneqq \int_{x-\frac\eta 2}^{x+\frac\eta 2} \ket{q}\!\bra{q}_{\hat q_\theta} \mathrm d q.
    \end{equation}
    This witness operator satisfies $\mathbb 0\preceq\hat W_{\theta,x,\eta}\preceq \mathbb 1$, and its threshold value at energy $E$ is defined by 
    \begin{equation}
    \label{eq:threshold_value_def_app}
        w_{\theta,x,\eta}^E\coloneqq\inf_{\hat\sigma \in \mathcal S_{0}^E}  \Tr(\hat\sigma\hat W_{\theta,x,\eta}).
    \end{equation}
We now consider the two functions $f: \eta  \mapsto \Tr\bigl(\hat\rho \hat W_{\theta,x,\eta}\bigr)$ and $w_{\theta,x}^E: \eta  \mapsto w_{\theta,x,\eta}^E$ in the limit $\eta  \to 0$. To prove the soundness (Theorem 4) and completeness (Theorem 5) of our witness in Definition~\ref{def:witness}, we will show that $f$ and $w_{\theta,x}^E$ have a different scaling in this limit. We start with the former. Let $p_{\hat\rho}(q)$ be the quadrature distribution of the target state $\hat\rho$ which by assumption has a zero at $x$. Here and in the following we omit the explicit reference to the corresponding quadrature angle $\theta$. The quadrature distribution  $p_{\hat\rho}(q)$ of any quantum state is a continuous, (at least) 1-differentiable function in $L^1(\mathbb R)$. Taylor's theorem tells us that 
\begin{equation}
    p_{\hat\rho}(q) =  p_{\hat\rho}'(x)(q-x) + R(q)(q-x),
\end{equation}
with $\lim_{q \to x}R(q) = 0$. 
With this we can write 
\begin{equation}
\begin{aligned}
     f(\eta ) &= \int_{x -\frac\eta 2}^{x +\frac\eta 2} p_{\hat\rho}(q) \mathrm dq \\
     &= \left|\int_{x -\frac\eta 2}^{x +\frac\eta 2} p_{\hat\rho}(q) \mathrm dq \right| \\
     &\leq  |p'_{\hat\rho}(x)|  \int_{x -\frac\eta 2}^{x +\frac\eta 2} |(q-x)|\mathrm dq  + \left(\max_{q \in [x -\frac\eta 2,x +\frac\eta 2]}|R(q)| \right) \int_{x -\frac\eta 2}^{x +\frac\eta 2}  |(q-x)|  \mathrm dq  \\ 
     &= \left(|p'_{\hat\rho}(x)|+\max_{q\in [x -\frac\eta 2,x +\frac\eta 2]}|R(q)|  \right) \int_{x -\frac\eta 2}^{x +\frac\eta 2}  |(q-x)|  \mathrm dq \\
     &= \left(|p'_{\hat\rho}(x)|+\max_{q\in [x -\frac\eta 2,x +\frac\eta 2]}|R(q)|  \right) \frac{\eta ^2}{4}.
     \end{aligned}
\end{equation}
The remainder term $R(q)$ from Taylor's theorem satisfies
\begin{align}
   \lim_{q \to x} R(q) = 0 \Leftrightarrow \quad   &  \lim_{\eta \to 0} R(x+\frac{\eta}{2}) = \lim_{\eta \to 0} R(x-\frac{\eta}{2}) = 0 \\
     \Leftrightarrow \quad  &  (\forall \epsilon >0) (\exists \eta_1)( \forall 0<\eta \leq \eta_1): |R(x\pm \frac{\eta}{2} )|\leq \epsilon.  
\end{align}
Consequently, there exists a $\eta_1 >0$ and some constant $c_1 > 0$ such that for all $0< \eta \leq \eta_1$ it holds that $f(\eta) \leq c_1 \eta^2$. Now we turn to the threshold value function $w_{\theta,x}^E$. It is defined by  
\begin{equation}
       w_{\theta,x}^E(\eta ) \coloneqq \inf_{\sigma \in \mathcal S_{0}^E}  \Tr(\sigma \hat W_{\theta,x,\eta}).
\end{equation}
 By our compactness result (Lemma~6) the infimum is contained in the set $\mathcal S_0^E$ and hence
\begin{equation}
\begin{aligned}
    \label{ineq:lowerboundgauss}
     w_{\theta,x}^E(\eta ) &= \min_{\sigma \in \mathcal S_{0}^E}  \Tr(\sigma \hat W_{\theta,x,\eta}) \\
    &= \min_{\sigma \in \mathcal S_{0}^E} \int_{x -\frac\eta 2}^{x +\frac\eta 2} p_\sigma(q)\mathrm dq \\
    & \geq \eta \min_{\sigma \in \mathcal S_{0}^E} \min_{q \in [x -\frac\eta 2,x +\frac\eta 2]} p_\sigma(q) .
\end{aligned}
\end{equation}
  The set $ \mathcal S_0^E$ corresponds to the set of statistical mixtures of energy-bounded, pure Gaussian states. Their quadrature distribution is strictly positive. Therefore for any finite value of $\eta$ we have
  \begin{equation}
\begin{aligned}
      w_{\theta,x}^E(\eta )  > 0. 
\end{aligned}
\end{equation}
The positivity of $p_\sigma(q)$ also implies that there exists a constant $c_2 > 0$ such that 
\begin{equation}
  \min_{\sigma \in \mathcal S_{0}^E} \min_{q \in [x -\frac\eta 2,x +\frac\eta 2]}  p_\sigma(q) \geq c_2,
\end{equation} for all $ \eta > 0$. 
Consequently it holds that $ w_{\theta,x}^E(\eta ) \geq c_2 \eta$.
Finally we notice that for any two constants $c,c' > 0$ there exists a $\eta_0 >0$ such that for all  $0 <\eta \leq \eta_0 $ it holds that $c' \eta > c\eta^2$. Putting everything together this means that we can find a $\delta =\min(\eta_1, \eta_0) > 0$ such that for all $0< \eta  \leq \delta$ it holds that
\begin{equation}
    0 < \Tr \bigl(\hat\rho \hat W_{\theta,x,\eta} \bigr) < w_{\theta,x,\eta}^E. 
\end{equation}

\subsection{State-independent guarantees}\label{sec:theo_guar_bin_size} 

The bounds obtained in the previous section show that for any state with a quadrature zero, there exists a window size such that the witness defined by the corresponding window around that zero successfully distinguishes the state from Gaussian states of the same energy. However, the exact value of the window size may depend on the specific unknown, tested state.

In this section, we prove stronger bounds, obtaining a state-independent value $\eta$ for the window size which only depends on the zero and the energy parameters. These bounds require an exponential energy bound on the tested state, which is a mild assumption satisfied by all states of finite stellar rank \cite{zeros}.

This results in a procedure that, given a quadrature angle, with high probability, detects a zero on a given interval of the quadrature values if it exists, or rejects if there is none. This is achieved by coarse-graining the interval of quadrature values into bins of size $\eta$, and testing all corresponding witnesses. In \cref{sec:heuristic_witness_construction}, we illustrate this procedure with a numerical study, further detailing how the confidence for the witness threshold violation can be optimized over the choice of bins.

We now turn to the results of this section: given a state $\psi$ with a quadrature zero at $q_\theta=x_0$ and satisfying an exponential energy bound $\bra{\psi}s^{\hat n}\ket{\psi}=S<+\infty$ for some $s>1$, we first derive an upper bound on the probability to obtain a quadrature outcome $x$ near its zero $x_0$, which only depends on $x_0$ and the energy parameters $s,S$ and not on the specific state $\psi$.

\begin{theo}\label{th:state_indpt_guar}
Let $\ket{\psi}$ be a quantum state with finite exponential energy $S = \bra{\psi}s^{\hat n}\ket{\psi}$ for some $s > 1$. Let $x_0\in\mathbb R$ and let $0<\eta\le1$ be a small interval width. Assuming that $\psi_\theta(x_0)=0$, the probability of finding the state within the interval $[x_0 - \eta/2, x_0 + \eta/2]$ satisfies the non-asymptotic bound:
\begin{equation}
    \Pr_{\psi}\left(x \in \left[x_0 - \frac{\eta}{2}, x_0 + \frac{\eta}{2}\right]\right) \le \eta^2 f(x_0, s, S),
\end{equation}
where $f(x_0, s, S)$ depends only on $x_0$ and the energy parameters $s,S$. 

\end{theo}

\begin{proof}
     
Suppose, up to applying a phase shift---which does not affect the energy bound---that $\theta = 0$. 

Let $x_0$ be such that $\psi(x_0) = 0$, where for a state 
$\ket{\chi}$, $\chi(x)$ denotes the position wavefunction evaluated at $x$. 
Then, if one denotes $\ket{\phi} =\hat D(-x_0)\ket{\psi}$, one has $\phi(0) = 0$, and for any $1 < t < s$, following \cite{zeros}:

\begin{equation} 
    \label{eq:displacedexpenergybound} 
    \bra{\phi}t^{\hat n}\ket{\phi} = \bra{\psi}D(x_0)t^{\hat n} D(-x_0)\ket{\psi} \le \frac{s}{t} \exp\left( \frac{(t-1)(s-1)}{s-t} \frac{x_0^2}{2} \right) S,
\end{equation} 
where $S = \bra{\psi}s^{\hat n}\ket{\psi}$. 

Then, by \cite[Lemma 9]{zeros}, one has
\begin{equation}\label{eq:expholobound}
    \vert \phi(x) \vert^2 \le K e^{L\vert x \vert ^2}, 
\end{equation}
with 
\begin{equation}\label{eq:KandL}
    K = \frac{C(t)^2 \sqrt{t}}{(\sqrt{t}-1)\sqrt{\pi}} \bra{\phi}t^{\hat n}\ket{\phi} \quad \text{and} \quad L = 1 + \frac{2}{e} + \frac{8}{t^\alpha - 1},
\end{equation}
for an arbitrary choice of $\alpha \in (0, 1/2)$, where $C(t) = \sup_{p \in \mathbb{N}} \sqrt{2p+1}\, t^{-(\frac{1}{2}-\alpha)p}$. 
By \cite[Theorem 1]{zeros}, the position 
wavefunction $\phi$ extends to a holomorphic function over the complex plane. Thus, 
it satisfies Cauchy's estimate, i.e.\ for any $r > 0$, for all $z \in \mathbb{C}$:

\begin{equation} 
    \vert \phi'(z)\vert \le \max_{w \in B(z, r)} \vert \phi(w)\vert \frac{1}{r}.
\end{equation} 

Let us denote $p_{\phi}(x) = \vert \phi(x)\vert^2$ the position probability distribution associated to $\phi$. 
Then, one has: 
\begin{align} 
    \vert p_\phi'(x) \vert &= \vert 2\Re(\phi(x)\overline{\phi'(x)})\vert \\ 
    &\le 2 \vert \phi(x) \vert \vert\phi'(x)\vert \\ 
    &\le 2 \vert \phi(x) \vert \max_{z \in B(x, r)}\vert \phi(z)\vert \frac{1}{r}\\ 
    &\le \frac{2K^2}{r} e^{L(\vert x\vert^2 + \vert x + r|^2)},
\end{align} 
where we used \cref{eq:expholobound} in the last line.
Since $p_{\phi}(0) = 0$, we deduce that for all $x \in [-\eta/2, \eta/2]$, $p_{\phi}(x) \le \frac{K^2\eta}{r}e^{L(\eta^2/4 + (\eta/2 + r)^2)}$.
In particular, integrating over $[-\eta/2, \eta/2]$ for some window size $\eta > 0$, one gets: 

\begin{align} 
    \Pr_{\psi}(x \in [x_0 - \eta/2, x_0 + \eta/2]) &= \Pr_{\phi}(x \in [-\eta/2, \eta/2]) \\ 
    &= \int_{-\eta/2}^{\eta/2} p_{\phi}(x)dx\\ 
    &\le \frac{K^2\eta^2}{r }e^{L(\eta^2/4 + (\eta/2 + r)^2)}\\
    &\le\eta^2(2K^2e^{2L}),
\end{align} 
where we have set $r=\frac12$ and used $\eta\le1$ in the last line.
To conclude, we set $f(x_0,s,S):=2K^2e^{2L}$ which only depends on $x_0$, $s$ and $S$ by combining \cref{eq:displacedexpenergybound} and \cref{eq:KandL}.

\end{proof}

\noindent We note that the factor $f(x_0,s,S)$  can be optimized by choosing the best values for the free parameters $r,t,\alpha$ appearing in the proof above and expressions of $K$ and $L$.

\smallskip

We now strengthen the lower bound (\ref{ineq:lowerboundgauss}) for energy-constrained Gaussian states.

\begin{lemma}\label{lem:lowerboundGauss}
Let $E\ge0$, let $x_0\in\mathbb R$ and let $0<\eta\le1$ be a small interval width. If $w_{x_0, \eta}^E$ denotes the minimum probability for a state in $\mathcal{S}_0^E$ to be measured at position $x \in [x_0 - \eta / 2, x_0 + \eta/2]$, then: 
    \begin{equation}
        w_{x_0, \eta}^E \ge\eta\,g(x_0,E),
    \end{equation}
where $g(x_0,E)$ only depends on $x_0$ and the energy $E$.         
\end{lemma}

\begin{proof}
By linearity of the probability density with $\rho \in \mathcal{S}_0^E$, the minimum must be obtained on a pure Gaussian state. Denote $\mu = \langle \hat q \rangle_\rho, \mu_p = \langle \hat p \rangle_\rho$ and $ \sigma^2 = \mathrm{Var}_\rho(\hat q^2) = \langle{\hat q^2}\rangle_\rho - \mu^2, \sigma_p^2 = \mathrm{Var}_\rho(\hat p) = \langle \hat p^2\rangle_\rho - \mu_p^2$, so that $\langle \hat n\rangle_\rho = \frac{1}{2}\left(\langle \hat q^2\rangle_\rho + \langle \hat p^2 \rangle_\rho \right) = \frac{1}{2} \left(\sigma^2 + \mu^2 + \sigma_p^2 + \mu^2\right)$. We also denote $\sigma_{qp} = \frac{1}{2}\langle \hat q \hat p + \hat p \hat q\rangle_\rho$. The energy constraint is then
\begin{equation}
    \sigma^2 + \mu^2 + \sigma_p^2 + \mu^2 \le 2E.
\end{equation}

Since $\rho$ is a pure Gaussian state, its position probability distribution at $q$ is
\begin{equation}
    p(q) = p_q(\mu, \sigma) = \frac{1}{\sqrt{2\pi}\sigma} \exp\left( -\frac{(q-\mu)^2}{2\sigma^2} \right).
\end{equation}
The derivative with respect to a position displacement is
\begin{equation}
    \frac{\partial p_q}{\partial \mu} = \frac{(q-\mu)}{\sigma^2} p_q(\mu, \sigma), 
\end{equation}
which vanishes only at the unique global maximum $\mu = q$. Thus, under any constraint of the form $\mu^2 + \sigma^2\leq C$, it is necessary to saturate the inequality to minimize $p$.

To minimize $p_q$ with fixed $\sigma$, we thus maximize $|\mu|$ within the energy constraint. From the constraint $\mu^2 = 2E - \sigma^2 - (\mu_p^2 + \sigma_p^2)$, this requires minimizing the momentum contributions. We can set $\mu_p = 0$. By the uncertainty principle $\sigma^2 \sigma_p^2 - \sigma_{qp}^2 \geq 1/4$, the minimum of $\sigma_p^2$ is $1/(4\sigma^2)$, reached when $\sigma_{qp} = 0$. Substituting these values, the optimal position displacement saturating the energy constraint is
\begin{equation}
    \mu = - \mathrm{sgn}(q) \sqrt{2E - \sigma^2 - \frac{1}{4\sigma^2}}.
\end{equation}
We now maximize the exponent and obtain a uniform bound: $(q - \mu)^2 \leq 2(q^2 + 2E) = C_{q,E}$. This reduces the problem to studying the objective function $G(\sigma) = \frac{1}{\sqrt{2\pi}\sigma} \exp( -\frac{C_{q,E}}{2\sigma^2})$.
Letting $x = \sigma^2$, the function $G(x) \propto x^{-1/2} e^{-C_{q,E}/2x}$ strictly increases until its unique maximum and then strictly decreases. Hence, the minimum of $G(\sigma)$ over the interval $[\sigma_{min}, \sigma_{max}]$ is attained at one of the endpoints. 
We conclude that the minimum probability density is reached at either the maximal compression $\sigma_{min}$ or the maximal expansion $\sigma_{max}$ allowed by the energy $E$, where $\sigma_{\pm}^2 = E \pm \sqrt{E^2 - 1/4}$.
A blunt approximation consists in taking $\sigma_+ \le \sqrt{2E}$ in the normalisation, and $\sigma_-^2 \ge E(1-\sqrt{1/4E}) \ge 1/8)$ in the exponent, giving
$p_q(\mu, \sigma) \ge \frac{1}{\sqrt{4\pi E}} e^{-8(q^2 + E)}$.

We finally obtain
\begin{align}
        w_{x_0, \eta}^E &\geq \eta \min_{\rho \in \mathcal S_{0}^E} \min_{q \in [x_0 -\frac\eta 2,x_0 +\frac\eta 2]} p_\rho(q) \\
        &\ge \frac{\eta}{\sqrt{4\pi E}} e^{-8((|x_0| + \frac{\eta}{2})^2 + E)}\\
        &\ge\eta\left(\frac{1}{\sqrt{4\pi E}} e^{-8((|x_0| + \frac12)^2 + E)}\right).
\end{align}
To conclude, we set $g(x_0,E):=\frac{1}{\sqrt{4\pi E}} e^{-8((|x_0| + \frac12)^2 + E)}$, which only depends on $x_0$ and $E$.
\end{proof}

Combining the two bounds from \cref{th:state_indpt_guar} and \cref{lem:lowerboundGauss}, we obtain that under an exponential energy condition $\langle s^{\hat n}\rangle\le S$, the witnesses $\hat W_{x_i,\eta,\theta}$ (where $x_i$ are centers of equally spaced bins of size $\eta$ in an interval $[-X,X]$) detect quantum non-Gaussianity against Gaussian states of energy $E$ if there is a quadrature zero in the interval $[-X,X]$ as long as
\begin{equation}
    \eta < g(X, E)/ f(X, s, S),
\end{equation}
independently of the tested state.

\subsection{Operational interpretation of witness violation}

Here we quickly prove Corollary 7 from the main text. For convenience we  reproduced it here:
\begin{cor}[Operational witness]
 Let $x\in\mathbb R$, $\eta >0$ and $E\ge0$. For any state $\hat\rho$ the amount of violation for the witness $\hat W_{\theta,x,\eta}$ satisfies
    \begin{equation}
        \langle \hat \Delta^E_{\theta,x,\eta } \rangle \equiv w_{\theta,x,\eta}^E-\Tr(\hat\rho\hat W_{\theta,x,\eta})\le \inf_{\hat\sigma\in\mathcal S_0^E} D(\hat\rho,\hat\sigma),
    \end{equation}
    where $D$ denotes the trace distance.
\end{cor}

\begin{proof}
The proof follows closely that of \cite[Lemma 1]{chabaud2021witnessing}. 
Consider the binary measurement $\{\hat W_{\theta,x,\eta},\mathbb 1- \hat W_{\theta,x,\eta}\}$ and the associated probability distribution $P^{\hat\sigma}_W(0)=1-P^{\hat\sigma}_W(1)= \mathrm{Tr}(\hat W_{\theta,x,\eta}\hat\sigma)$ for a state $\hat\sigma$. 
Let $\hat \sigma\in\mathcal S_0^E$ so that $\mathrm{Tr}(\hat W_{\theta,x,\eta}\hat\sigma)\ge w_{\theta,x,\eta}^E$, by definition of the threshold value. We have:
\begin{equation}
    \begin{aligned}
       w_{\theta,x,\eta}^E -\Tr(\hat W_{\theta,x,\eta}\hat\rho)&\le\left|\Tr(\hat W_{\theta,x,\eta}\hat\rho)-\Tr(\hat W_{\theta,x,\eta}\hat\sigma)\right|\\
        &= \|P^{\rho}_W-P^\sigma_W\|_\mathrm{tvd}\\
        &\le D(\hat \rho,\hat \sigma),
    \end{aligned}
\end{equation}
where $\|\cdot\|_\mathrm{tvd}$ denotes the total variation distance and the last line follows from the operational property of the trace distance. Taking the infimum over $\hat \sigma\in\mathcal S_0^E$ concludes the proof.
\end{proof}

\section{Explicit expressions for computing witness threshold values}
\label{sec:explicit_expressions_threshold}

To perform the optimization in the definition of the threshold value (see Eq.~\eqref{eq:threshold_value_def}) in practice we provide the expressions for the quadrature distribution and average energy for any pure state $\ket{\psi_r}$ of finite stellar rank $r$ in terms of the parametrization \cite{chabaud_stellar_2020}
\begin{equation}
    \ket{\psi_r} = \hat D(\alpha) \hat S(\chi)\sum_{n =0}^r c_n \ket{n}.
\end{equation}
The quadrature distribution of $\ket{\psi_r}$ along $\hat q_\theta = e^{-i \theta \hat n }\hat x e^{i \theta \hat n } $ is
\begin{equation}
    |\langle q_\theta | \psi_r\rangle |^2 = \left| \sum_{n=0}^r c_n \frac{e^{in\theta}}{\sqrt{2^n n!}} H_n\left( \frac{q_\theta - \mu_\theta(\alpha)}{\sqrt{2}\sigma_\theta(\chi)} \right) \left[ \frac{\cosh{|\chi|} - e^{-i(\phi + 2\theta)}\sinh{|\chi|}}{\sqrt{2}\sigma_\theta(\chi)} \right]^n \right|^2 \, |\langle q_\theta | \alpha, \chi \rangle|^2. \label{eq:quad_distr_*rank_r_state}
\end{equation}
Here $H_n(x)$ is the $n^{\text{th}}$ Hermite polynomial, $\chi = |\chi|e^{i\phi}$, and $|\langle q_\theta | \alpha, \chi \rangle|^2$ is the quadrature distribution of the Gaussian state $\ket{\alpha, \chi}$,
\begin{equation}\label{DSVac_quad}
    |\langle q_\theta | \alpha,\chi\rangle|^2 = \frac{1}{\sigma_\theta(\chi)\sqrt{2\pi}} e^{-\frac{1}{2}\left( \frac{q_\theta - \mu_\theta(\alpha)}{\sigma_\theta(\chi)}\right)^2},
\end{equation} 
with mean and standard deviation
\begin{equation}
\begin{aligned}
    \mu_\theta(\alpha) &= \sqrt{2}\text{Re}(\alpha e^{-i\theta}) \\
    \sigma_\theta(\chi) &= \sqrt{\frac12( \cosh(2|\chi|) - \cos(\phi + 2\theta)\sinh(2|\chi|) )}.
\end{aligned}
\end{equation}
For the mean energy of the state we get
  \begin{multline}\label{eq:energy_SRr}
         \bra{\psi_r}\hat n \ket{\psi_r} = |\alpha|^2 + \sinh^2(|\chi|) + \cosh(2|\chi|)\sum_{n=1}^r n |c_n|^2 - 2  \sinh(|\chi|)\cosh(|\chi|)\sum_{n=0}^{r-2}\sqrt{n+1}\sqrt{n+2} \text{Re}\bigl(e^{-i\phi}c_n^*c_{n+2}\bigr) \\
      +\cosh(|\chi|)\sum_{n=0}^{r-1}\sqrt{n+1} \text{Re}\bigl(\alpha^*c_n^*c_{n+1}\bigr)-\sinh(|\chi|)\sum_{n=0}^{r-1}\sqrt{n+1} \text{Re}\bigl(\alpha e^{-i \phi }c_n^*c_{n+1}\bigr),
    \end{multline}
where as usual a summation with a smaller upper limit then the initial value is defined to be zero.  
All above expressions can be obtained by direct manipulation \cite{Moller_DSN_states_position_1996}. Eq.~\eqref{eq:quad_distr_*rank_r_state} in fact represents the explicit form of the result from Theorem~3 in our companion paper \cite{zeros} where we abstractly prove that the wave function of a state with stellar rank $r$ is the product of a complex Gaussian with a complex polynomial of degree $r$.

We now give more detail on the low-dimensional examples discussed in the main text.  If the target state is the Fock state $\ket{1}$, it is rotation-symmetric and so all quadrature distributions $|\langle q_\theta | 1 \rangle|^2 = \frac{2}{\sqrt{\pi }} q_\theta^2 e^{-q_\theta^2}$ have the same zero-set, namely $q_\theta=0$.
The set of pure states with stellar rank 0
is exactly the set of pure single-mode Gaussian states $\ket{\alpha, \chi}$ with $\chi = |\chi|e^{i\phi} \in \mathbb C$ and $\alpha \in \mathbb C$.

Without loss of generality in this rotation-symmetric example, we restrict to the position quadrature $\hat q_{\theta=0}= \hat x$.  The target state expectation is
\begin{equation}\label{Fock1witnessValue}
\begin{aligned}
    \langle 1 | \hat W_{0,0,\eta } |1 \rangle \! &= \!\int_{-\frac\eta 2}^{\frac\eta 2} |\langle x | 1\rangle|^2 \mathrm d x = - \frac{\eta}{\sqrt{\pi}} e^{- (\frac{\eta }{2})^2} + \erf\left( \frac{\eta }{2} \right),
\end{aligned}
\end{equation}
while the unbounded squeezed coherent states yield
\begin{equation}\label{pureGaussian_on_Fock1Witness}
\langle \alpha,\chi| \hat W_{0,0,\eta } | \alpha,\chi\rangle = \int_{-\frac\eta 2}^{\frac\eta 2} |\langle x | \alpha,\chi\rangle|^2 \mathrm d x = \frac{1}{2}\left[ \erf\left( \frac{\frac{\eta }{2}-\mu_0(\alpha)}{\sqrt{2}\sigma_0(\chi)} \right) - \erf\left( \frac{-\frac{\eta }{2}-\mu_0(\alpha)}{\sqrt{2}\sigma_0(\chi)} \right)  \right] ,
\end{equation}
where $\erf(z)$ is the error function.  See Fig.\ 1 in the main text (blue solid line) for a plot of the optimized witness violation value $\langle \hat \Delta^E_{\theta,x,\eta } \rangle$ from Corollary 7 as a function of $\eta$.  In the case of the lossy target state $\hat \rho(p) = p\ketbra{0}{0}+(1-p)\ketbra{1}{1}$, the target state expectation changes linearly (for a fixed $\eta$) as a function of the loss $p$:
\begin{equation}
    \tr[\hat \rho(p) \hat W_{0,0,\eta}] = p\int_{-\frac\eta 2}^{\frac\eta 2} |\langle x | 0\rangle|^2 \mathrm d x + (1-p)\int_{-\frac\eta 2}^{\frac\eta 2} |\langle x | 1\rangle|^2 \mathrm d x = \langle 1| \hat W_{0,0,\eta } | 1\rangle + p\frac{\eta}{\sqrt{\pi}} e^{- (\frac{\eta }{2})^2},
\end{equation}
See the inset of Fig.\ 1 for this dependence, both in the single-quadrature case and the many-quadrature case.

Now consider pure states of stellar rank 1, $\ket{\psi_1} = \hat D(\alpha) \hat S(\chi) (\cos\frac{s}{2}\ket{0} + e^{it}\sin\frac{s}{2}\ket{1})$. Specifying \cref{eq:quad_distr_*rank_r_state}, its quadrature distribution is 
\begin{align}
    |\langle q_\theta |\psi_1\rangle|^2 &= \Bigg[ \cos^2\frac{s}{2} + \frac{q-\mu_\theta(\alpha)}{\sqrt{2}\sigma_\theta(\chi)^2} \sin s \Big(\cosh r \cos (\theta +t)-\sinh r \cos (t - \theta - \phi )\Big) + \left(\frac{q-\mu_\theta(\alpha)}{\sigma_\theta(\chi)} \right)^2 \sin ^2 \frac{s}{2} \Bigg] |f_0(q_\theta)|^2,
\end{align}
and its energy (\cref{eq:energy_SRr}) reduces to
\begin{align}
    \langle \psi_1 | \hat a^\dagger \hat a | \psi_1 \rangle &= \sinh^2 (|\chi|) +|\alpha|^2 + \sin^2(\frac{s}{2}) \cosh(2|\chi|)  + |\alpha| \sin(s) \Big( \cos(t - \theta_\alpha)\cosh(|\chi|) - \cos(t - \phi + \theta_\alpha)\sinh(|\chi|) \Big).
\end{align}
Using the notion $z_- = \frac{x - \frac{\eta }{2}-\mu_\theta(\alpha)}{\sigma_\theta(\chi)}$ and $z_+ = \frac{x + \frac{\eta }{2}-\mu_\theta(\alpha)}{\sigma_\theta(\chi)}$, the unbounded witness expectation is computed as
\begin{equation}
\begin{aligned}
    \int_{x-\frac{\eta}{2}}^{x+\frac{\eta}{2}} |\langle q_\theta | \psi_1 \rangle |^2 dq &= \frac{1}{2} \left[ \erf\left( \frac{z_+}{\sqrt{2}} \right) - \erf\left( \frac{z_-}{\sqrt{2}} \right)  \right] + \sin^2 \frac{s}{2} \frac{1}{\sqrt{2 \pi }}\left(z_- e^{-\frac{1}{2}z_-^2} - z_+ e^{-\frac{1}{2}z_+^2}\right) \\
    &\quad + \sin s \Big(\cosh r \cos (\theta +t)-\sinh r \cos (t - \theta - \phi )\Big) \frac{1}{\sigma_\theta(\chi)2\sqrt{\pi}} \left[  e^{-\frac12 z_-^2} - e^{-\frac12 z_+^2}\right].
\end{aligned}
\end{equation}
So if the target state is the second Fock state, $\ket{2}$, this wavefunction has two zeroes at $\pm \frac{1}{\sqrt{2}}$ along all quadratures.  Hence, if one only uses a single quadrature, then the minimization program (as a function of $\eta$) for the threshold is
\begin{align}
    &\min_{\alpha,\chi,s,t} \left[ \int_{-\frac{1}{\sqrt{2}}-\frac{\eta}{2}}^{-\frac{1}{\sqrt{2}}+\frac{\eta}{2}} |\langle q_\theta | \psi_1 \rangle |^2 dq + \int_{\frac{1}{\sqrt{2}}-\frac{\eta}{2}}^{\frac{1}{\sqrt{2}}+\frac{\eta}{2}} |\langle q_\theta | \psi_1 \rangle |^2 dq \right], \quad \text{such that } \langle \psi_1 | \hat a^\dagger \hat a | \psi_1 \rangle \leq 2.
\end{align}
In this rather special case of a single quadrature on the Fock-2 target state, we may without loss of generality set $\theta=0$.  For a more general target state, we would need to find the privileged angles and focus only on them as described in the main text. 

\section{Generalised witnesses}

Here we discuss two alternative versions of our primary witness \eqref{eq:witness_def} and explore their consequences on the witness-ability of non-Gaussianity.
 
\subsection{Gaussian windows}
\label{sec:Gaussian_windows}
Here we investigate a modification of the witness operator \eqref{eq:witness_def} based on a Gaussian blurring of the quadrature POVM.  Let $g_{x,\gamma}(\cdot)$ be the normalized Gaussian probability distribution with mean $x$ and standard deviation $\gamma$.  The Gaussian-blurred quadrature operators are
\begin{equation}
    \hat V_{\theta,x, \gamma} := \int g_{0,\gamma}(q - x) \ket{q}\!\bra{q}_{\hat q_\theta} \, dq = \int_\mathbb{R} g_{x,\gamma}(q) \, \ket{q}\!\bra{q}_{\hat q_\theta} \, dq.
\end{equation}
This convolves each quadrature POVM element $\ket{q}\!\bra{q}_{\hat q_\theta}$ with a Gaussian of width $\gamma$, which plays the role of an effective ``homodyne resolution''.  We then coarse-grain over a bin of size $\eta$ centered on $x$,
\begin{equation}
    \hat{W}_{\theta,x,\gamma,\eta}^\mathrm{soft} = \int_{x-\eta/2}^{x+\eta/2}\hat{V}_{\theta,y,\gamma}dy,
\end{equation}
corresponding to a soft homodyne binned measurement.  For a given quadrature distribution $|\psi(q_\theta)|^2$ and $\eta>0$, this new witness first convolves $|\psi(q_\theta)|^2$ by a Gaussian with size $\gamma$ (i.e.\ a form of classical added white noise) then integrates over the window. Note that, similar to the original witness \eqref{eq:witness_def}, this is a POVM element for $\eta >0$.

The Gaussian structure of the convolution allows for obtaining convenient expressions for the resulting distributions. Following the main example from the text, here we explore Fock state $\ket1$ as a target state. The soft analogue of Eq.~\eqref{Fock1witnessValue} can be found to be
\begin{equation}
\begin{aligned}
    \langle 1 | \hat W^{\text{soft}}_{0,0,\gamma,\eta } |1 \rangle \! &= -\frac{\eta}{\sqrt{\pi } \left(2 \gamma ^2+1\right)^{\frac{3}{2}}} e^{-\frac{\eta ^2}{4(2\gamma ^2 + 1)}} + \text{erf}\left(\frac{\eta }{2 \sqrt{2 \gamma ^2+1}}\right),
\end{aligned}
\end{equation}
while the soft analogue of the pure Gaussian value, Eq.~\eqref{pureGaussian_on_Fock1Witness}, generalizes to 
\begin{equation}
\langle \alpha,\chi| \hat W^{\text{soft}}_{0,0,\gamma,\eta } | \alpha,\chi\rangle = \frac{1}{2}\left[ \erf\left( \frac{\frac{\eta }{2}-\mu_0(\alpha)}{\sqrt{2(\sigma^2_0(\chi) + \gamma^2)}} \right) - \erf\left( \frac{-\frac{\eta }{2}-\mu_0(\alpha)}{\sqrt{2(\sigma^2_0(\chi) + \gamma^2)}} \right)  \right].
\end{equation}
In both cases $\lim_{\gamma \rightarrow 0}\langle \hat W^{\text{soft}}_{0,0,\gamma,\eta } \rangle = \langle \hat W_{0,0,\eta } \rangle$, as expected.  With these we then numerically compute the optimal energy-bounded Gaussian expectation value and determine any threshold violation as a function of $\gamma$ (see \cref{fig:3quadsoft}).

\begin{figure}[h!]
    \centering
    \includegraphics[width=0.5\linewidth]{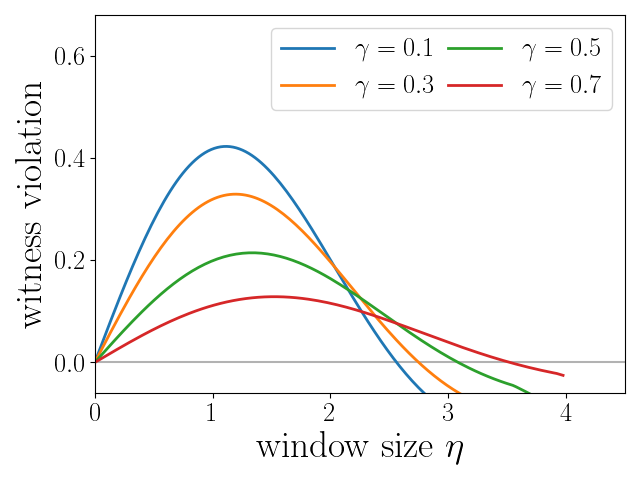}
    \caption{Non-Gaussian witness violation with a Fock $\ket{1}$ target using three quadratures at angles along the three roots of unity, and the soft homodyne binning as a function of the standard deviation $\gamma$.  Smaller values of $\gamma$ increase the maximal witness violation as expected.}
    \label{fig:3quadsoft}
\end{figure}

Similar to the original witness \eqref{eq:witness_def}, we expect that soundness and completeness theorems can be established for the witness $\hat{W}_{\theta,x,\gamma,\eta}^\mathrm{soft}$, for small enough values of $\gamma>0$. 

\subsection{Extended witnesses for certification without energy promise}
\label{sec:no_energy_promise}

To certify non-Gaussianity from measurements of a single quadrature, an energy promise on the state to certify is required, as discussed in the main text.  In the following, we consider an extension of the witness operator $\hat W_{\theta,x,\eta}$ from Eq.~(\ref{eq:witness_def}) which allows to drop the required energy promise on the target state  at the cost of leaving the realm of a single quadrature measurement. Let $\hat W$ be the projector from Definition~\ref{def:witness} and $p \in [0,1]$. 
We define \begin{equation}\label{eq:gen_witness_def}
    \hat V_{\bm{\theta},\bm{x},\eta}^{(p)} \coloneqq p \, \hat n + (1-p) \hat W_{\bm{\theta},\bm{x},\eta},
    \end{equation}
    where $\hat n$ is the photon number operator, and the corresponding threshold value  
    \begin{equation}
\label{eq:gen_threshold_value_def}
    v_{\bm{\theta},\bm{x},\eta}^{(p)}\coloneqq\inf_{\hat\sigma \in \mathcal S_{k-1}}  \Tr(\hat\sigma  \hat V_{\bm{\theta},\bm{x},\eta}^{(p)}).
    \end{equation}
    In this construction, the energy of the adversarial state will compete with its ability to minimize its probability mass in the windows around the zeros of the target state. 
    The advantage of this extended witness is that we can perform the optimization involved in determining the threshold value now over the set of $\mathcal S_{k-1}$ instead of $\mathcal S_{k-1}^E$, lifting the constraint on the energy of the states we optimize over.  
    This allows us to adopt the strategy from \cite{fiuravsek2022efficient} to determine the threshold value. 
    The insight we draw from this work is that one can rewrite the threshold value as an optimization of the minimal eigenvalue of a finite-dimensional matrix. Denoting by $\Pi_{k-1}= \sum_{n=0}^{k-1} \ketbra{n}{n}$ the projector onto the space spanned by the first $k$ Fock states, it is shown in \cite{fiuravsek2022efficient} that
    \begin{equation}
        v_{\bm{\theta},\bm{x},\eta}^{(p)} = \inf_{\alpha,\chi}\quad \min \text{eig} \bigl(\Pi_{k-1} \hat S^\dagger(\chi) \hat D^\dagger (\alpha) \hat V_{\bm{\theta},\bm{x},\eta}^{(p)} \hat D(\alpha) \hat S(\chi)\Pi_{k-1}\bigr).
    \end{equation}
When witnessing stellar rank $k$, the dimension of this matrix is $k$ and its diagonalization can be efficiently carried out numerically. 
In this altered approach, the number of parameters involved in the final optimization is independent of the stellar rank $k$ one seeks to certify, as opposed to our original approach where the number of parameters scales linearly with $k$. But there is a trade-off here: Reducing the number of parameters we optimize over increases the computation time of the cost function value at a fixed set of parameters. So it is not directly clear that the computation of $v_{\bm{\theta},\bm{x},\eta}^{(p)}$ is more efficient that the computation of the original threshold.
    
We remark that the extended witness requires measuring at least two conjugate quadratures in order to evaluate the expectation value of the witness operator in Eq.~(\ref{eq:gen_witness_def}). This clearly takes us beyond the realm of certifying stellar rank from single-quadrature measurements. However, as discussed in the main text, witnesses built from more than one quadrature will likely be necessary in practice regardless, due to the presence of noise.
To demonstrate the completeness of the extended witness numerically, we consider again the Fock state $\ket{1}$ as the target state to be certified against all Gaussian states. We take the original witness $\hat W_{\bm\theta,\bm x, \eta}$ built from windows around the zero at two conjugated quadratures. The measurement of those two conjugated quadratures is sufficient to estimate the expectation value of the full extended witness (which additionally involves an energy estimation).
The results for the respective threshold value and expected value of the extended witness in the Fock state $\ket{1}$ are plotted in Fig.~\ref{fig:extended_witness_num}.
\begin{figure}
    \centering
    \includegraphics[width=0.5\linewidth]{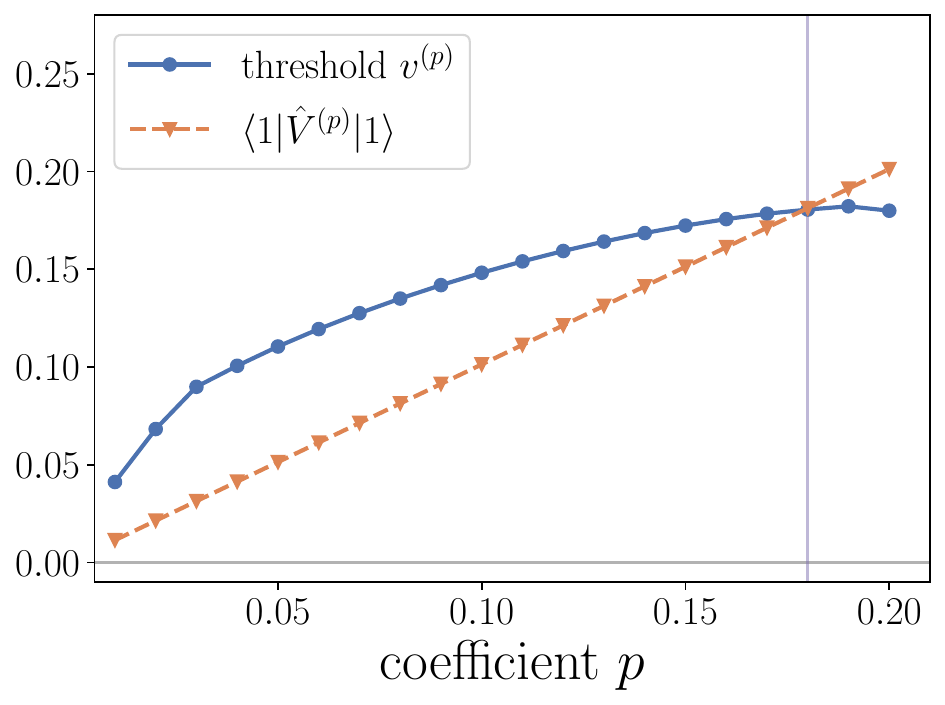}
    \caption{Numerical example for the extended witness operator $\hat V^{(p)}$ having the Fock state $\ket{1}$ as a target state. The original witness operator used is $\hat W_{0,0, 0.2}+ \hat W_{\tfrac \pi 2,0, 0.2}$, which is built from the projectors onto the window of width $\eta =0.2$ around the zero along two conjugated quadratures. Plotting the respective threshold value (blue, circles) and the expectation of the extended witness in the target state (orange, triangles) for different convex combinations, we see that for $p< 0.18$, the generalized witness can detect the non-Gaussianity of the Fock state $\ket{1}$.}
    \label{fig:extended_witness_num}
\end{figure}
We see that for $p< 0.18$, the extended witness can detect the non-Gaussianity of the Fock state $\ket{1}$.

\section{Witnessing higher stellar ranks}
\label{sec:witness_stellar}
  To certify stellar rank from measurements of a single quadrature against all states in $\mathcal S_{k-1}^E$, the procedure based on the witness in Definition~\ref{def:witness} can be extended if $k$ zeros are present in a quadrature distribution for some $\theta$. As before, we assume an energy promise on our target state. Our witness is then generalized to be the projector onto small windows around all the $k$ zeros of a quadrature distribution of a given target state. Consequently, the respective threshold is the infimum of the witness expectation value over all states in $\mathcal S_{k-1}^E$. In our companion paper \cite{zeros} we point out more in detail that the wavefunction of a stellar rank $r$ state has exactly $r$
  complex zeros. Consequently each state in $\mathcal S_{k-1}^E$ can have at most $k-1$ real zeros in its wavefunction. This can be used to show, along the same lines as in the proof of Theorems 4 and 5, that there exists in this more general case a non-trivial window size $\eta$ for the witness which ensures soundness and completeness of the witnessing strategy. Moreover, the violation of a witness threshold value also admits the operational interpretation as in Corollary 7: it lower bounds the trace distance between the measured state which is an element in $\mathcal S_{k}^E$ and all states in $\mathcal S_{k-1}^E$.

One can then ask, if arbitrary high stellar rank can in principle be certified from measurements of a single quadrature. 
Our results in the companion paper (reproduced as Theorem 1 in the main text of the present paper) state that pure states of infinite stellar rank satisfying the energy condition $\langle s^{\hat n}\rangle_\psi<+\infty$ for some $s>1$ have entire wave functions with infinitely many isolated complex zeros. 
In particular, this implies that our non-Gaussianity witnesses can be applied to witness arbitrarily high stellar rank for states whose wave function has infinitely many real zeros along some quadrature angle by picking windows around any number of such zeros. 
Cat states $\ket{\psi_\alpha} \propto \ket{\alpha} + \ket{-\alpha}$, with $\ket{\alpha}= \hat D(\alpha) \ket{0}$ being a displaced vacuum state, satisfy $\langle s^{\hat n}\rangle_{\psi_\alpha} = C_{\alpha} e^{-|\alpha|^2}\cosh(|\alpha|^2\sqrt{s})< +\infty$ for any fixed $s > 0$ and some normalization constant $C_\alpha$.   
The cat state's quadrature distribution along any angle $\theta$ is
\begin{equation}
    |\bra{\psi_\alpha} \!q\rangle_{\hat q_\theta}|^2 \propto e^{-(q-\sqrt{2}\Re(\alpha e^{-i\theta}))^2}+e^{-(q+\sqrt{2}\Re(\alpha e^{-i\theta}))^2}+ 2e^{-q^2-2\Re(\alpha e^{-i\theta})^2}\cos(2\sqrt{2}\Im(\alpha e^{-i\theta})q).
\end{equation}
We see from the cosine that the cat state has infinitely many real, isolated zeros along any quadrature $\hat q_\theta$ such that $\Im(\alpha e^{-i\theta}) \neq 0$. 
Thus, the cat state is a suitable target state for which arbitrarily high stellar rank can be witnessed by our strategy from measurements of a single quadrature.

\section{Sample complexity for multiple quadratures}
\label{sec:sampleComplexityManyQuads}

In this section we show a simple generalization of Theorem~8 in the main text to witnesses built using more than one quadratures. In that case, given points $\bm x=(x_1,\dots,x_n)$, bin sizes $\bm\eta=(\eta_1,\dots,\eta_N)$ and quadrature angles $\bm\theta=(\theta_1,\dots,\theta_N)$, the corresponding witness is given by
\begin{equation}\label{eq:witnessmultiquad}
    \hat{W}_{\bm{\theta},\bm{x},\bm{\eta}}\coloneqq\sum_{j=1}^N\hat{W}_{\theta_j,x_j,\eta_j}.
\end{equation}
Its Gaussian threshold value at energy $E$ is defined as
\begin{equation}\label{eq:witnessthreshmultiquad}
    w_{\bm{\theta},\bm{x},\bm{\eta}}^E\coloneqq\min_{\hat\sigma\in\mathcal S_0^E}\mathrm{Tr}\!\left[\hat\sigma\hat{W}_{\bm{\theta},\bm{x},\bm{\eta}}\right].
\end{equation}

Suppose $M=M_1+\dots+M_N$ copies of a state $\hat\sigma$ are available. The witnessing protocol is then as follows: for all $j=1,\dots,N$, measure $\hat q_{\theta_j}$ for $M_j$ copies of $\hat\sigma$, and compute the fraction $\bar{w}_{\hat\sigma,j}$ of outcomes falling in the range $[x_j-\frac{\eta_j}2,x_j+\frac{\eta_j}2]$. Then, compute the experimental violation of the threshold value
\begin{equation}
    w_{\bm{\theta},\bm{x},\bm{\eta}}^E-\sum_{j=1}^N\bar{w}_{\hat\sigma,j},
\end{equation}
and declare the state $\hat\sigma$ as non-Gaussian if this  quantity is positive. Due to statistical uncertainty, if the violation or the number of samples is too small, a Gaussian state can still pass this test with some nonzero probability.
The following result bounds the failure probability of this witnessing protocol, i.e.\ the probability of incorrectly labeling the state $\hat\sigma$ as non-Gaussian when it is Gaussian.

\begin{theo}[Sample complexity for multiple-quadrature witnesses]\label{thm:sample_complexity_multi}
    Consider a witness $\hat{W}_{\bm{\theta},\bm{x},\bm{\eta}}$ as in Eq.~\eqref{eq:witnessmultiquad} with threshold $w_{\bm{\theta},\bm{x},\bm{\eta}}^E$ at energy $E$. Suppose that the state being measured is in $\mathcal{S}_0^E$ (i.e.\ it is Gaussian with energy bounded by $E$). The probability $p_\mathrm{fail}$ that the estimator derived from $M_1+M_2+\ldots+M_N$ measurement outcomes for the quadrature angles $\theta_1,\theta_2,\ldots,\theta_N$ will be smaller than a fixed $w_{\bm{\theta},\bm{x},\bm{\eta}}^E-\epsilon$ with $\epsilon>0$ satisfies
    \begin{equation}\label{eq:pfail_multi}
        p_\mathrm{fail}\leq \sum_j\exp\left( -2M_j\left(\frac{\epsilon}{N}\right) ^2\right).
    \end{equation}
\end{theo}

\begin{proof}
    The proof is similar to that of Theorem 8 with the difference that we apply Hoeffding's inequality separately to the samples obtained from measuring different quadratures. To do so, we consider 
    \begin{equation}
        \mathrm{Tr}\left[\hat\sigma \hat{W}_{\bm{\theta},\bm{x},\bm{\eta},}  \right] = \sum_j  \mathrm{Tr}\left[\hat\sigma \hat{W}_{\theta_j,x_j,\eta_j}  \right].
    \end{equation} 
    We can estimate this quantity in an experiment by estimating each term on the right hand side individually, using $M_1$ i.i.d.\ samples for the quadrature at angle $\theta_1$, $M_2$ at angle $\theta_2$, and so on. We denote the corresponding estimates by $\bar{w}_{\hat\sigma,j}\in[0,1]$ for $j=1,\dots,N$. While the samples in each batch are i.i.d., this is not necessarily the case between different batches, so Hoeffding's inquality does not directly apply. However, note that
    \begin{align}
        \mathrm{Pr}\left[\mathrm{Tr}\left[\hat\sigma \hat{W}_{\bm{\theta},\bm{x},\bm{\eta}}\right]-\sum_{j=1}^N\bar{w}_{\hat\sigma,j}>\epsilon\right]&=\mathrm{Pr}\left[\sum_{j=1}^N\left(\mathrm{Tr}\left[\hat\sigma \hat{W}_{\theta_j,x_j,\eta_j}\right]-\bar{w}_{\hat\sigma,j}\right)>\epsilon\right]\\
        &\leq\mathrm{Pr}\left[\exists j,\mathrm{Tr}\left[\hat\sigma \hat{W}_{\theta_j,x_j,\eta_j}\right]-\bar{w}_{\hat\sigma,j}>\frac{\epsilon}N\right]\\
        &\leq\sum_{j=1}^N\mathrm{Pr}\left[\mathrm{Tr}\left[\hat\sigma \hat{W}_{\theta_j,x_j,\eta_j}\right]-\bar{w}_{\hat\sigma,j}>\frac{\epsilon}N\right],
    \end{align}
    where we used the union bound in the last line.
    Applying Hoeffding's inequality to upper bound each probability on the right hand side of the last expression we obtain
    \begin{equation}\label{eq:HoeffdingNquad}
        \mathrm{Pr}\left[\mathrm{Tr}\left[\hat\sigma \hat{W}_{\bm{\theta},\bm{x},\bm{\eta}}\right]-\sum_{j=1}^N\bar{w}_{\sigma,j}>\epsilon\right]\le\sum_{j=1}^N\exp\left(-2M_j\left(\frac{\epsilon}{N}\right)^2\right),
    \end{equation}
    valid for all states $\hat\sigma$.

    Let us now assume that $\hat\sigma\in\mathcal S_0^E$. In that case,
    \begin{equation}
        \mathrm{Tr}\left[\hat\sigma \hat{W}_{\bm{\theta},\bm{x},\bm{\eta}}  \right]\ge w_{\bm{\theta},\bm{x},\bm{\eta}}^E,
    \end{equation}
    by definition of the witness threshold value. Hence, with Eq.~\eqref{eq:HoeffdingNquad}, we obtain
    \begin{align}
        p_\mathrm{fail}&=\mathrm{Pr}\left[\sum_{j=1}^N\bar{w}_{\hat\sigma,j} < w_{\bm{\theta},\bm{x},\bm{\eta}}^E- \epsilon\right]\\
        &\le\mathrm{Pr}\left[\mathrm{Tr}\left[\hat\sigma \hat{W}_{\bm{\theta},\bm{x},\bm{\eta}}\right]-\sum_{j=1}^N\bar{w}_{\hat\sigma,j}>\epsilon\right]\\
        &\le\sum_{j=1}^N\exp\left(-2M_j\left(\frac{\epsilon}{N}\right)^2\right).
    \end{align}
\end{proof} 

Note that the operator in Eq.~(\ref{eq:witnessmultiquad}) is not a POVM element since in general $\hat{W}_{\bm{\theta},\bm{x},\bm{\eta}} \npreceq \mathbb{1}$, but this can be achieved by dividing it by $N$. We chose the former convention to simplify the proof (in particular the application of Hoeffding's inequality) but normalizing to a POVM it is also possible to extend the operational meaning of Lemma~7 in the main text to the general multi-quadrature setting. 

\section{Heuristic approach to the choice of the optimal bin from measurement data}
\label{sec:heuristic_witness_construction}

\begin{figure*}[t]
\centering
\includegraphics[width=\textwidth]{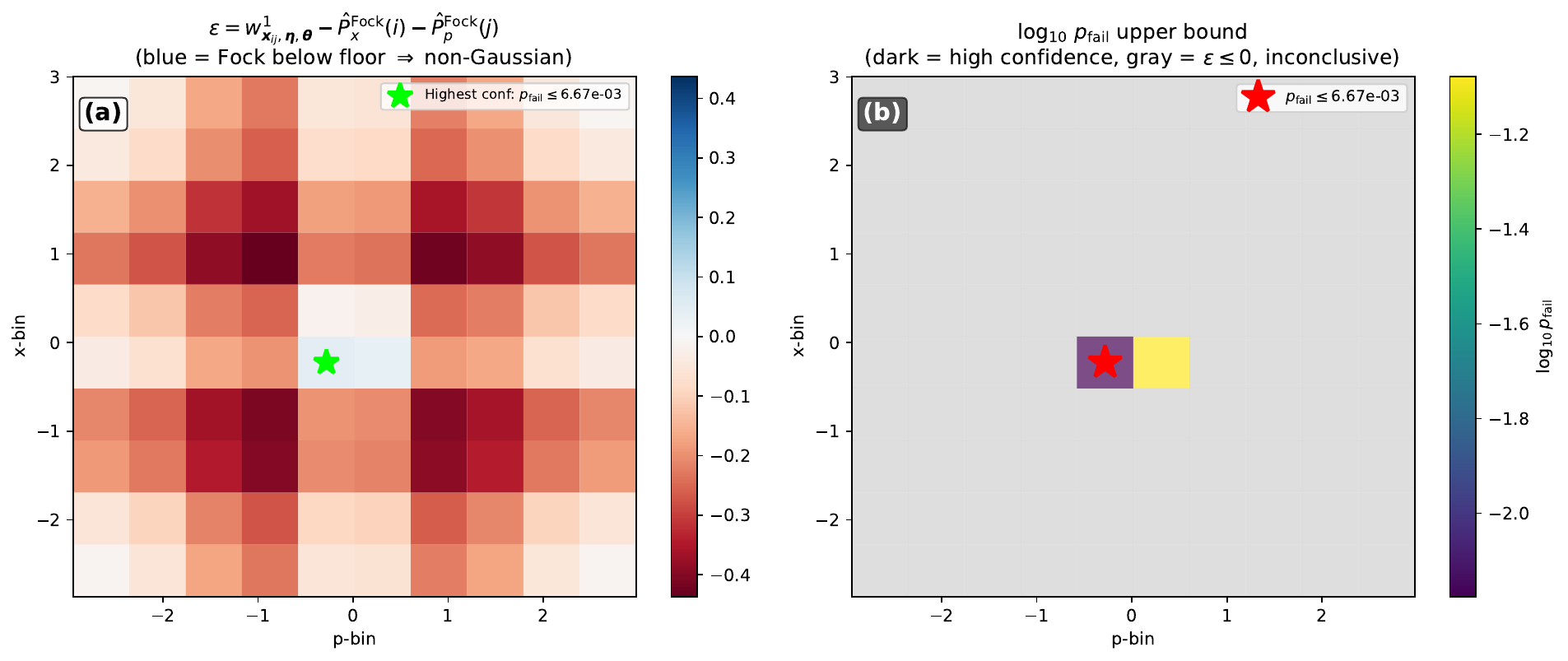}
\caption{%
Non-Gaussianity witness for Fock state $\ket{1}$ versus all Gaussian
states with $\langle \hat{n}\rangle \leq 1$. We consider witness operators $\hat{W}_{ij} = \hat{W}_{\bm{\theta},\boldsymbol{x}_{ij},\boldsymbol{\eta}}$ as defined in Eq.~\eqref{eq:witnessmultiquad} with $\boldsymbol{x}_{ij} = (x_i,p_j)$  spanning a two dimensional grid (10 bins per quadrature), a constant width $\boldsymbol{\eta} = \left(\eta,\eta\right) $, so that squares tessellate the observed quadrature ranges, and $\boldsymbol{\theta} = \left(0,\pi/2\right)$.
(a)~Difference $\epsilon = w^1_{\boldsymbol{x}_{ij},\boldsymbol{\eta},\boldsymbol{\theta}} - \hat{P}_x^{\mathrm{Fock}}(i) - \hat{P}_p^{\mathrm{Fock}}(j)
$ across all bin pairs $(i,j)$; Here $\hat{P}_x^{\mathrm{Fock}}(i)$ is the frequency of outcomes falling within the position bin centered at $x_i$ in a batch of $M_x = M_p = M = 5000$ position measurements on independent input states, assumed to be single-photon states, and similarly for momentum; 
blue regions indicate the estimated values for the witness on a Fock state falling \emph{below} the Gaussian
floor. (b)~Upper bound $p_\mathrm{fail}\leq e^{-2M_x (\epsilon/2)^2} + e^{-2M_p (\epsilon/2)^2}$ on a
$\log_{10}$ scale; darker regions correspond to smaller values (stronger
evidence); gray cells have $\epsilon \leq 0$ (no witness violation).  The red
star marks the best (smallest) bound on $p_\mathrm{fail}$.
}
\label{fig:witness_selection}
\end{figure*}

In order to construct our quadrature witnesses we need to assume some prior knowledge of the state, namely which quadratures to measure to observe a zero and how to bin the outcomes of these quadratures' measurements so that the corresponding witness shows a violation. Here we outline a heuristic approach to choose the latter depending on measurement data, provided it is known that the state to certify has zeros in the measured quadratures. This approach is guaranteed to succeed when there is a zero in the considered interval of quadrature values thanks to the state-independent theoretical guarantees in \cref{th:state_indpt_guar} proven in \cref{sec:theo_guar_bin_size}.

To be concrete, we consider the example of witnessing quantum non-Gaussianity (as opposed to witnessing higher stellar rank). As noted above, in general some knowledge of the state is needed to choose the quadrature(s) to measure in order to reveal that the corresponding distribution has a zero, but let us consider for simplicity an experiment to certify a phase-invariant state, making the choice of quadratures irrelevant \footnote{If needed, this symmetry can be tested with a third batch of measurements.}, and a witness defined using two quadratures: position and momentum.

One may divide the samples for each quadrature in two batches: the first batch is used to reconstruct the quadrature density and produce a histogram; this is used to choose the window location and size around a dip in the reconstructed probability density; finally, the second batch is used to (independently) estimate the corresponding witness value.

The two phases of the experiment would run as follows. 

\textbf{Phase 1.} In the first phase, a certain number of copies of the input state are used to sample from the selected quadratures. The quadrature samples are then used to construct histograms with a pre-determined number of bins, say $B$ for each quadrature. The ranges for the histograms can be chosen to match the range of the samples for the respective quadratures. The observed frequencies for each pair of bins $\left(\Delta_{q,j},\Delta_{p,j}\right)$ are then used to estimate the value of the witness corresponding to that bin. This value is compared with the minimum value that can be attained by a Gaussian. The latter step is the most computationally expensive, as it entails a separate optimization for each of the $B^2$ bins. In our implementation~\cite{data_plotting} we use grid search with a successive refinement, which takes approximately $10$ seconds for $B\approx 100$ on a laptop.

\textbf{Phase 2.} Based on the previous phase, the bin with the best violation confidence is chosen, and more quadrature measurements are sampled from fresh copies of the state to be certified. Sampling the quadratures on new copies ensures statistical independence, since using the same data to select the witness and evaluate it might introduce bias undermining the certification (in particular the confidence bounds in \cref{thm:sample_complexity_multi}).

The results of an example run of phase 1 with a (pure) single-photon state are displayed in \cref{fig:witness_selection}. The main figure of merit for the non-Gaussianity certification is the upper bound in Eq.~\eqref{eq:pfail_multi}. \cref{fig:confidence_convergence} shows how this bound improves with larger number of samples in phase 2. In particular, this analysis reveals that a moderate number of measurements are sufficient to witness non-Gaussianity with high confidence (around $10^4$ samples for $99\%$ confidence).
Importantly, while prior knowledge about the tested state can significantly improve the success rate of our witnesses, a full theoretical model of the tested state is not strictly necessary.


\begin{figure}[t]
\centering
\includegraphics[width=\columnwidth]{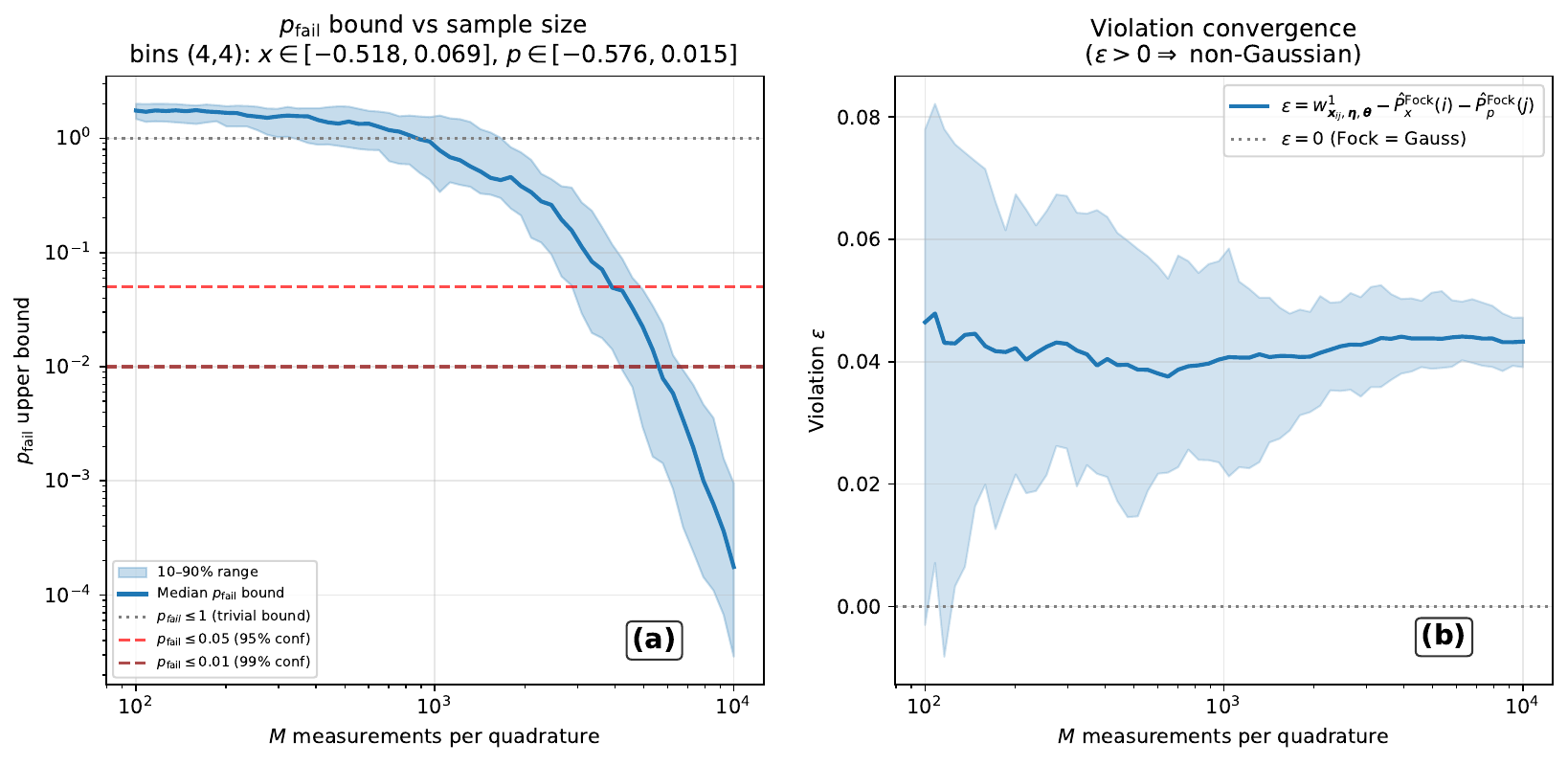}
\caption{%
Convergence of the non-Gaussianity witness with measurement number for
the optimal bin pair identified in Fig.~\ref{fig:witness_selection}.
(a)~Upper bound to the failure probability $p_\mathrm{fail}\leq e^{-2M_x (\epsilon/2)^2} + e^{-2M_p (\epsilon/2)^2}$
as a function of the number of homodyne measurements $M$ per quadrature
($M_x = M_p = M$).  The solid line shows the median over 20 independent
repetitions.  Horizontal dashed lines
indicate $p_\mathrm{fail} \leq 1 $ (no evidence), $p_\mathrm{fail} \leq  0.05$, and $p_\mathrm{fail} \leq 0.01$.
(b)~Witness violation $\epsilon = w^1_{\boldsymbol{x}_{ij},\boldsymbol{\eta},\boldsymbol{\theta}} - \hat{P}_x^{\mathrm{Fock}}(i) - \hat{P}_p^{\mathrm{Fock}}(j)$ stabilizing at positive
values as $M$ increases, demonstrating that the Fock state genuinely
violates the Gaussian bound.
}
\label{fig:confidence_convergence}
\end{figure}

\end{document}